\documentclass[authoryear,12pt]{article}

\usepackage{algorithm}
\usepackage{algorithmic}

\usepackage[english]{babel}
\usepackage{array}
\usepackage{rotating}
\usepackage{color}
\usepackage{setspace}
\usepackage{url}
\usepackage{multirow}
\usepackage{epstopdf}
\usepackage{float}
\usepackage{subcaption}
\usepackage{myCmdShortcutPaperI}
\usepackage{myCmdShortcutPaperII}
\usepackage{myCmdShortcutPaperII_Extra}

\RequirePackage{amsthm,amsmath,amsfonts,amssymb}

\usepackage{natbib}
\usepackage{mathrsfs}
\RequirePackage[colorlinks,citecolor=blue,urlcolor=blue]{hyperref}
\RequirePackage{graphicx}
\usepackage{enumitem}

\newcounter{magicrownumbers}

\RequirePackage{geometry}
\RequirePackage{t1enc}
\RequirePackage{enumerate}
\RequirePackage{epsfig}
\RequirePackage{makeidx}
\RequirePackage{graphicx}
\RequirePackage{listings}
\RequirePackage{gensymb}

\usepackage[utf8]{inputenc}
\usepackage[utf8]{luainputenc}
\usepackage[english]{babel}

\usepackage{xcolor}

\makeindex
\newcommand{\namelistlabel}[1]{\mbox{#1}\hfil}

\newtheorem{definition}{Definition}[section]
\newtheorem{remark}{\bf Remark}[section]

\newtheorem{lemma}{\bf Lemma}[section]
\newtheorem{theorem}{\bf Theorem}[section]
\newtheorem{result}{\bf Result}[section]
\newtheorem{illustration}{\bf Illustration}[section]

\theoremstyle{remark}

\usepackage[symbol]{footmisc}

\begin{document}

	\thispagestyle{empty}

\pagenumbering{arabic}
\newpage


\title{G-optimal grid designs for kriging models }
\author{{\small Subhadra Dasgupta$^1$, Siuli Mukhopadhyay$^{2}$ and Jonathan Keith$^3$}\\
	{\small \it $^1$ IITB-Monash Research Academy, India}\\
	{\small \it $^2$ Department of Mathematics, Indian Institute of Technology Bombay, India}\\
	{\small \it $^3$ School of Mathematics, Monash University, Australia}}
\date{}

\maketitle

\begin{abstract}
This work is focused on finding G-optimal designs theoretically for kriging models with two-dimensional inputs and separable exponential covariance structures. For design comparison, the notion of evenness of two-dimensional grid designs is developed. The mathematical relationship between the design and the supremum of the mean squared prediction error ($SMSPE$) function is studied and then optimal designs are explored for both prospective and retrospective design scenarios. In the case of prospective designs, the new design is developed before the experiment is conducted and the regularly spaced grid is shown to be the G-optimal design. The retrospective designs are constructed by adding or deleting points from an already existing design. Deterministic algorithms are developed to find the best possible retrospective designs (which minimizes the $SMSPE$). It is found that a more evenly spread design under the G-optimality criterion leads to the best possible retrospective design. For all the cases of finding the optimal prospective designs and the best possible retrospective designs, both frequentist and Bayesian frameworks have been considered. The proposed methodology for finding retrospective designs is illustrated with a methane flux monitoring design.		


\end{abstract}

\section{Introduction}

Efficient designing of experiments plays an important role in the area of geo-statistics, as it has been established by many researchers that sampling designs have an impact on the prediction accuracy of geostatistical processes. This article is motivated by such problems where the input space is two-dimensional. For example, measuring methane emission rates may depend on inputs such as temperature and atmospheric pressure. \cite{baran2015optimal} finds the optimal designs for such experiments. \cite{Cokriging_Li_Zimmerman_2015} discuss the problem of finding the optimal design for monitoring soil quality data on a river basin, where the inputs are two-dimensional spatial locations. \cite{mavukkandy2014assessment} discuss a problem of analyzing river water quality variables, in this case, the inputs could be taken as two-dimensional spatial locations or one-dimensional space along with time. 

In this article, we work with isotropic separable exponential processes. The covariance of the random  process $Z(\cdot)$ is given by $\displaystyle{Cov(Z(x, y), Z(x^{\prime}, y^{\prime})) = \sigma^{2} \; {e}^{-\alpha|x-x^{\prime}|} \; {e}^{-\beta|y-y^{\prime}|}}$, where $\displaystyle{\sigma^{2}, \alpha,  \text{ and } \beta >0}$ are the covariance parameters. The focus is to theoretically find the prospective optimal grid designs and the best possible retrospective grid designs for these random processes, obtained by minimizing the supremum of kriging variance (\textit{SMSPE}), hereafter called the {\em G-optimal designs}. For the case of \textit{prospective designs}, in which the design is set up before data collection begins, we theoretically show that equispaced grids are G-optimal. In the case of \textit{retrospective designs}, the new design has to be formed by adding/deleting points to/from an existing design \citep{Diggle_Bayesian_Geostatistical_Design}. This kind of case could arise in various ongoing spatial or spatio-temporal experiments. To find the best possible retrospective designs, a criterion to compare the evenness of two-dimensional grids is proposed, using the definition of majorization of vectors. The focus of this article is to study the mathematical relation between the evenness of grid design and the $SMSPE$ function and then determine the best possible design (rather than determining the optimal design numerically by comparing values of $SMSPE$). For this reason, we consider only the separable exponential covariance structure for the random function; in fact, no theoretical exact optimal designs for covariance structures other than the exponential covariance are currently available in the statistical literature. Using the criterion for comparing designs based on the evenness of two-dimensional grids we provide converging and deterministic algorithms for finding the best possible retrospective designs. All the above optimal and best possible designs are determined under both frequentist and Bayesian paradigms. To illustrate the proposed retrospective designs we use an existing $8 \times 8$ methane flux monitoring grid design (used in \citealp{baran2015optimal}) and delete points to find the best possible $6 \times 5$ grid design.

The problem of theoretically finding prospective optimal designs for Ornstein-Uhlenbeck processes has been of interest to many authors for example, \cite{kiselak2008equidistant}, \cite{Kriging_Zagoraiou_Antognini_2009}, \cite{Kriging_Antognini_Zagoraiou_2010}, \cite{Baran_Stehlik_2015}, \cite{baran2013optimal}, \cite{baran2015optimal}, \cite{baran2017k},  \cite{sikolya2019optimal}, and \cite{Paper_One_2021_arxiv}. 
For Ornstein-Uhlenbeck (OU) processes with one-dimensional inputs, \cite{kiselak2008equidistant}, \cite{Kriging_Zagoraiou_Antognini_2009}, \cite{Kriging_Antognini_Zagoraiou_2010}, and \cite{Paper_One_2021_arxiv} found optimal designs under various criteria like D-optimality, $\text{D}_s$-optimality, entropy, and integrated mean square prediction error (\textit{IMSPE}). Very recently, \cite{Paper_One_2021_arxiv} determined G-optimal and pseudo-Bayesian G-optimal designs for bivariate cokriging models with one-dimensional inputs. 
For Ornstein-Uhlenbeck processes with two-dimensional inputs, \cite{baran2013optimal} and \cite{Baran_Stehlik_2015} theoretically found optimal designs for OU sheets over two-dimensional monotonic sets, under D-optimality, entropy, and \textit{IMSPE} criteria. \cite{baran2015optimal} proved that an equispaced two-dimensional grid design is optimal under the \textit{IMSPE} criterion. \cite{sikolya2019optimal}  worked with the prediction of complex Ornstein-Uhlenbeck processes and found theoretical optimal designs with respect to the entropy criterion.

Though the literature discussed above mainly aims to find theoretical prospective optimal designs under various design criteria, they do not consider the G-optimality criterion. Thus in this manuscript, we aim to theoretically determine G-optimal designs for kriging models with two-dimensional inputs. Until now, only numerical approaches have been used to find the G-optimal prospective designs for kriging with two-dimensional inputs (see \citealp{sacks1989design} , \citealp{van2000influence}, and \citealp{kriging_Zimmerman_2006}). In the case of retrospective spatial design, \cite{Diggle_Bayesian_Geostatistical_Design} investigated numerically Bayesian geostatistical designs for correlated processes over a two-dimensional space. \cite{jones1998efficient}, \cite{schonlau1998global}, and \cite{ranjan2008sequential} also gave iterative numerical algorithms involving a convergence criterion to determine such designs for kriging models.

In this article, for determining new retrospective designs by the addition of points, we discuss two broad cases. The first case involves sequentially adding the points (one point at a time). In this case, the proposed algorithm is theoretical and does not depend on the values of the covariance parameters. The second case {involves} adding the points {simultaneously}. For this case, the algorithm is converging and deterministic, but it is a combination of theory and computations. {To find} the best possible retrospective design by deleting points from an existing design, we {consider} the case of deleting all points simultaneously. The algorithm we provide for the deletion case is again a combination of theory and computations. We do not discuss the case of deleting points sequentially, as we were not able to obtain a purely theory-based algorithm (as we did for the addition of points). 

The main contribution of this paper, for kriging models where the random process is an OU sheet are: i) theoretically finding prospective G-optimal grid designs (under both frequentist and Bayesian setups), ii) proposing a criterion to compare the evenness of two design grids using majorization and then prove that in fact, a more evenly spread design will minimize the \textit{SMSPE}, and iii) providing deterministic algorithms to find the best possible retrospective designs for both cases of addition or deletion of points from an existing design (under both frequentist and Bayesian {frameworks}).

In Section \ref{Paper2_Model}, we discuss the model, covariance structure, and some notations. In Section~\ref{Paper_II_Design}, we state the design criteria. In Sections~\ref{Paper_II_Prospective_Design} and ~\ref{Paper_II_Retrospective_Design}, G-optimal prospective designs and best possible retrospective designs {are discussed}, respectively under both frequentist and Bayesian paradigms. In Section~\ref{Paper_II_Illustration}, we {present some case studies}. We compare the existing designs with various newly obtained designs (after the addition or deletion of points). Section~\ref{Paper_II_Conclusion} concludes the paper and provides some future directions.

\section{Linear model and preliminaries} \label{Paper2_Model}%
Consider a random function $\displaystyle{\{ Z(x, y): (x, y) \in \mathcal{D}_{1} \times \mathcal{D}_{2} \}} $, where $\mathcal{D}_{1}$ and $ \mathcal{D}_{2}$ are two real intervals ($\mathcal{D}_{1}, \mathcal{D}_{2} \subseteq \mathbb{R}$). The function $Z(\cdot)$ is sampled over a two-dimensional grid $\mathcal{S}   $ $ (\subseteq \mathcal{D}_{1} \times \mathcal{D}_{2})$, which is generated by a set $\mathcal{S}_{1} =  \{ x_{1}, \ldots, x_{n}  \} $ for the $x$-covariate and a set $\mathcal{S}_{2} =  \{ y_{1}, \ldots, y_{m}  \} $ for the $y$-covariate, {such that} $\displaystyle{\mathcal{S} = \mathcal{S}_{1} \times \mathcal{S}_{2} =  \{ (x_{i}, y_{j}): i = 1,\ldots, n \text{ and } j = 1,\ldots, m \}}$. The random process is characterized as
\begin{eqnarray}
	Z(x,y) &= \pi + \epsilon(x,y) \text{ for } (x,y) \in \mathcal{D}_{1} \times \mathcal{D}_{2},   
\end{eqnarray}
where, $E[\epsilon(x,y)] = 0 $, $\pi$ is the mean of the random function $Z(\cdot)$, and $Cov[\epsilon(x,y), \epsilon(x^\prime, y^\prime)] = \mathcal{C}((x,y),(x^\prime,y^\prime))$ for some isotropic covariance kernel $\mathcal{C}(\cdot,\cdot)$ and $(x,y), (x^\prime,y^\prime) \in \mathcal{D}_{1} \times \mathcal{D}_{2} $. In this article, we use a separable exponential covariance with parameters $\sigma^{2}, \alpha, \text{ and } \beta >  0 $. 
Therefore $\displaystyle{
	Cov(Z(x, y), Z(x^{\prime}, y^{\prime})) = \sigma^{2} \;  \CovFunctP(|x-x^{\prime}|) \;\CovFunctQ(|y-y^{\prime}|)}$, where  $\CovFunctP(|x-x^{\prime}|) =  {e}^{-\alpha|x-x^{\prime}|}$ and $\displaystyle{\CovFunctQ(|y-y^{\prime}|) = {e}^{-\beta|y-y^{\prime}|}}$ for $(x, y)$, $\displaystyle{(x^{\prime},y^{\prime}) \in \mathcal{D}_{1} \times \mathcal{D}_{2}}$. 

Let $\displaystyle{
	\VectZ = (Z(x_{1}, y_{1}), Z(x_{1}, y_{2}), \ldots	, Z(x_{1}, y_{m}),  \ldots,	Z(x_{n}, y_{1}), Z(x_{n}, y_{2}), \ldots	 ,Z(x_{n}, y_{m}))^{T}} $ and $\displaystyle{\VectEpsilon^{T} = (\epsilon(x_{1}, y_{1}), \epsilon(x_{1}, y_{2}), \ldots	, \epsilon(x_{1}, y_{m}),  \ldots,	\epsilon(x_{n}, y_{1}), \epsilon(x_{n}, y_{2}), \ldots	 ,\epsilon(x_{n}, y_{m}))^{T}}$.  
Then, in matrix notation the model is given by $ \VectZ  = \VectOneMN \pi + \VectEpsilon 
$, where the $mn \times 1$ vector $\VectOneMN = (1,1,\ldots,1)^{T}$. 

Our interest is in predicting $Z(\cdot)$ at $(x_{0}, y_{0}) \in \mathcal{D}_{1} \times \mathcal{D}_{2}$ using simple and ordinary kriging models. The kriging estimator of $Z(x_{0}, y_{0})$ is denoted by $Z^{\ast}(x_{0}, y_{0})$ and has the form  $\displaystyle{Z^{\ast}(x_{0}, y_{0}) = \sum_{i = 1}^{n}  \sum_{j = 1}^{m} \lambda_{ij} \; Z(x_{i}, y_{j}) + \lambda_{0}}$. It is the best linear unbiased predictor (\textit{BLUP}) (see Chapter 3 \citealp{Book_Chiles_Delfiner} and \citealp{Hoef_1993}). 

Some notations used throughout the paper are listed below. The random variable $\displaystyle{Z(x_0, y_0)}$ is denoted by $\displaystyle{Z_{0}}$, that is, $\displaystyle{Z(x_{0}, y_{0}) \equiv Z_{0}}$. The covariance vectors are $\displaystyle{Cov(Z_{0},Z_{0}) = \sigma_{00} }$, $\displaystyle{Cov(\VectZ, Z_{0}) =  \sigmaNot_{mn\times 1}   }$,  and $\displaystyle{Cov(\VectZ,\VectZ) = \pmb \Sigma_{mn\times mn}}$. 
The covariance matrices and vectors $\MatrixP$ and $\sigmaPNot$ correspond to $\mathcal{S}_{1}$, {whereas} $\MatrixQ$ and $\sigmaQNot$ correspond to $\mathcal{S}_{2}$. Therefore, in this case we write $\displaystyle{(\MatrixP)_{ij} = {e}^{-\alpha|x_{i}-x_{j}|}}$ and $\displaystyle{(\sigmaPNot)_{i} = {e}^{-\alpha|x_{i}-x_{0}|}}$ for all $\displaystyle{i,j= 1,\ldots,n}$, $\displaystyle{(\MatrixQ)_{ij} = {e}^{-\beta|y_{i}-y_{j}|}}$ and $\displaystyle{(\sigmaQNot)_{i} = {e}^{-\beta|y_{i}-y_{0}|}}$ for all $i,j= 1,\ldots,m$. This kind of separable covariance structure allows us to use the properties of {the} Kronecker product and leads to a simplification of covariance vectors and matrices as follows: 
\begin{eqnarray}
\covMat & = \sigma^{2} \; \MatrixP \otimes \MatrixQ  \label{Paper_II_MatrixDecomp_1}, \\
\sigmaNot &= \sigma^{2} \; \sigmaPNot  \otimes  \sigmaQNot \label{Paper_II_VectorDecomp_1}.
\end{eqnarray}
 \par In {the} case of simple kriging, we assume that the random function does not have a drift; equivalently, {the} mean of the random function (denoted by $\pi$) is constant and known. Without loss of {generality,} the random process could be written as $  Z(x,y) = \epsilon(x,y)$. This assumption of known mean is valid when the random process is observed repeatedly and the mean is estimated almost with perfection \citep[Chapter~ 3]{Book_Chiles_Delfiner}. The simple kriging estimator of $Z_{0}$ is given by $\displaystyle{Z^{\ast }_{sk}= \sigmaNot^{T} \covMat^{-1} \VectZ}$ and the kriging variance, which is the mean squared prediction error $(MSPE)$ at $(x_{0},y_{0})$ is given by $\displaystyle{ \sigma^{2}_{sk}(x_{0},y_{0})= \sigma_{00} - \sigmaNot^{T}  \covMat^{-1} \sigmaNot}$ {\citep[see][Chapter~3]{Book_Chiles_Delfiner}}. Using equation~\eqref{Paper_II_MatrixDecomp_1} and \eqref{Paper_II_VectorDecomp_1} we {obtain the following expression:}
\begin{align} 
Z^{\ast }_{sk} &= \Big(\sigmaPNot^{T}\MatrixP^{-1} \otimes  \sigmaQNot^{T}\MatrixQ^{-1}  \Big) \VectZ, \label{Paper_II_SimplekrigingEstimatorII} \\
\sigma^{2}_{sk}(x_{0},y_{0}) 
&= \sigma^{2} \Big( 1- \sigmaQNot^{T}  \MatrixQ^{-1} \sigmaQNot \;\; \sigmaPNot^{T}  \MatrixP^{-1} \sigmaPNot \Big). \label{Paper_II_SimplekrigingMSPEII} 
\end{align}

In the ordinary kriging setup, the mean of the random function (denoted by $\pi$) is a constant but unknown. From \citet[Chapter~ 3]{Book_Chiles_Delfiner}, the ordinary cokriging estimator of $Z_{0}$ is given by $\displaystyle{Z^{\ast }_{ok} = \sigmaNot^{T}  \SigmaBold^{-1}     \VectZ +  \dfrac{(1-\sigmaNot^{T} \SigmaBold^{-1} \VectOneMN)(\VectOneMN^{T} \SigmaBold^{-1} \VectZ)}{\VectOneMN^{T} \SigmaBold^{-1} \VectOneMN} }$ and kriging variance $(MSPE)$ is given by 
$\displaystyle{ \sigma^{2}_{ok}(x_{0},y_{0}) = 
	\sigma_{00} -  
	\begin{pmatrix}
	1 & \sigmaNot^{T}
	\end{pmatrix}
	\begin{pmatrix}
	0 & \VectOneMN^{T}\\
	\VectOneMN & \SigmaBold 
	\end{pmatrix}^{-1} 
	\begin{pmatrix}
	1 \\ \sigmaNot
	\end{pmatrix}   }$.
Again, using equations~\eqref{Paper_II_MatrixDecomp_1} and \eqref{Paper_II_VectorDecomp_1} we get the following simplified expressions: 
\begin{eqnarray} 
& Z^{\ast }_{ok} = \Big(  \sigmaPNot^{T}\MatrixP^{-1} \otimes \sigmaQNot^{T}\MatrixQ^{-1}  \Big) \VectZ 
+ \dfrac{\Big(1- \sigmaQNot^{T} \MatrixQ^{-1} \VectOneM \;\; \sigmaPNot^{T} \MatrixP^{-1} \VectOneN \Big) ( \VectOneN^{T}\MatrixP^{-1} \otimes \VectOneM^{T}\MatrixQ^{-1} )\VectZ}{ \VectOneM^{T} \MatrixQ^{-1} \VectOneM \;\; \VectOneN^{T} \MatrixP^{-1} \VectOneN}, \label{OrdinarykrigingEstimatorII} \\
& \sigma^{2}_{ok}(x_{0},y_{0}) 
= \sigma^{2} \Bigg( 1 - \sigmaQNot^{T}  \MatrixQ^{-1} \sigmaQNot \;\; \sigmaPNot^{T}  \MatrixP^{-1} \sigmaPNot + \dfrac{\Big(1-\sigmaQNot^{T} \MatrixQ^{-1} \VectOneM  \;\; \sigmaPNot^{T} \MatrixP^{-1} \VectOneN \Big)^2}{\VectOneM^{T} \MatrixQ^{-1} \VectOneM \;\; \VectOneN^{T} \MatrixP^{-1} \VectOneN }\Bigg). \label{OrdinarykrigingMSPEII}
\end{eqnarray}

\section{Optimal design} \label{Paper_II_Design}	

In practice, an experimenter comes across two situations for designing a network of sampling points: the first is the \textit{prospective design}, which involves choosing design points before experimentation{;} the second is the \textit{retrospective design}, which involves improving an existing network by adding new points or deleting sampling points \citep{Diggle_Bayesian_Geostatistical_Design}.  

The arrangement of sampling points is known as the design of the experiment. In this article, an $n\times m$ grid design denoted by $\mathcal{S}$, which is generated by $\mathcal{S}_{1}$ and $\mathcal{S}_{2}$ is taken (stated in Section~\ref{Paper2_Model}). For kriging models, extrapolation should be treated with caution (as noted by \citealp{sikolya2019optimal}), therefore $\mathcal{D}_{1} = [x_{1}, x_{n}]$ and $\mathcal{D}_{2}=  [y_{1}, y_{m}]$. Since replications are not allowed, assume that the points are ordered; $x_{i} < x_{i^{\prime}}$  and $y_{j} < y_{j^{\prime}}$ whenever $i < i^{\prime}$ and $j < j^{\prime}$, respectively. The partition sizes (distances between two consecutive points on axis) are given by $d_{i} = x_{i+1} - x_{i}$ and $\delta_{j} = y_{j+1} - y_{j}$ for $ i= 1,\ldots, n-1$ and $j = 1,\ldots, m-1$. The sets $\mathcal{S}_{1}$ and $\mathcal{S}_{2}$ could equivalently be denoted by the vectors $\VectD = ( x_{1}, d_{1},\ldots,d_{n-1})$ and $\VectDelta = ( y_{1}, \delta_{1},\ldots,\delta_{m-1})$, respectively. Note, that $ \displaystyle{x_{n} = x_{1} + \sum_{i = 1}^{n-1} d_{i}}	$ and $ \displaystyle{y_{m} = y_{1} + \sum_{j = 1}^{m-1} \delta_{j}} $. Hence, the grid design $\mathcal{S}$ could equivalently be denoted by $\xibold = (\VectD,\VectDelta)$. 

The objective is to choose a two-dimensional grid which maximizes the prediction accuracy of $Z(\cdot)$. Therefore the mean squared prediction error, $MSPE((x_{0}, y_{0}), \xibold, \CovParameterTheta)$ is minimized, where $\CovParameterTheta = (\sigma^2, \alpha, \beta)$ is the parameter vector corresponding to the exponential covariance. Since kriging is an interpolation technique, the prediction variance at the sampling points is zero. {Figure \ref{fig:SMSPE OK} shows variation in \textit{MSPE} over $\mathcal{D}_{1} \times \mathcal{D}_{2}$ for an ordinary kriging example.} At the sampling points on the (sampling) grid the $MSPE$ is exactly equal to zero and it increases as one moves away from the sampling points. For this reason, in this article, the supremum of $MSPE$ is used to find the optimal design (also discussed in Chapter 6, \citealp{Book_Santner}  and \citealp{pronzato2012design}). This optimization criterion ensures that points which have high variance are discouraged, since these points are given less prominence when working with an average of the $MSPE$.
\begin{figure}[H]
	\centering
	\begin{subfigure}{.5 \textwidth}
		\centering
		\includegraphics[width=.9\linewidth]{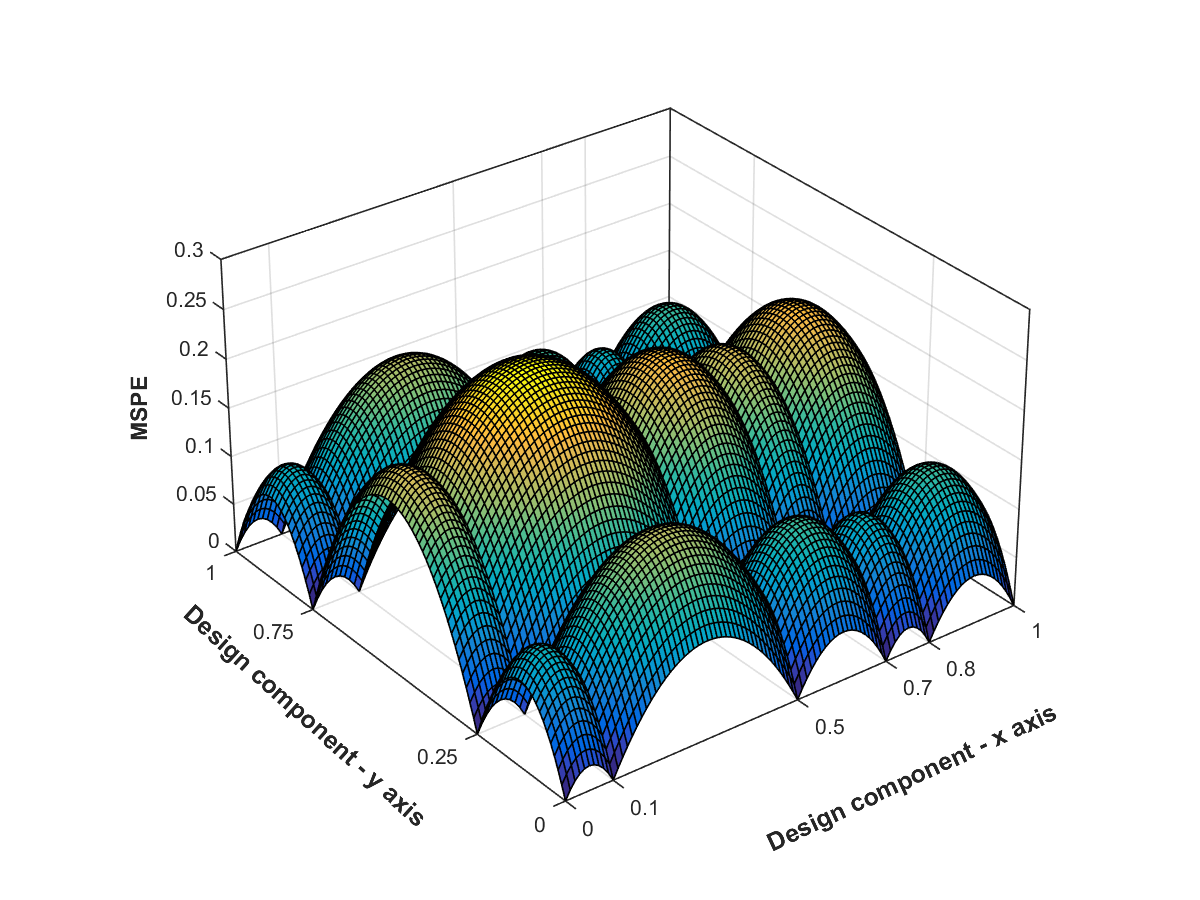}
			\captionsetup{justification=centering}
			\caption{\scriptsize $\mathcal{S}_{1} = \{0  ,  0.1 ,0.5   , 0.7,  0.8 ,   1 \}$ and $\mathcal{S}_{2} =  \{0 ,   0.25   ,  0.75  ,  1\}$. $\sigma^{2} = 1$, $\alpha = .5$, and $\beta = .7$.}  
		\label{fig:SMSPE OK I}
	\end{subfigure}%
	\begin{subfigure}{.5\textwidth}
		\centering
		\includegraphics[width=.9\linewidth]{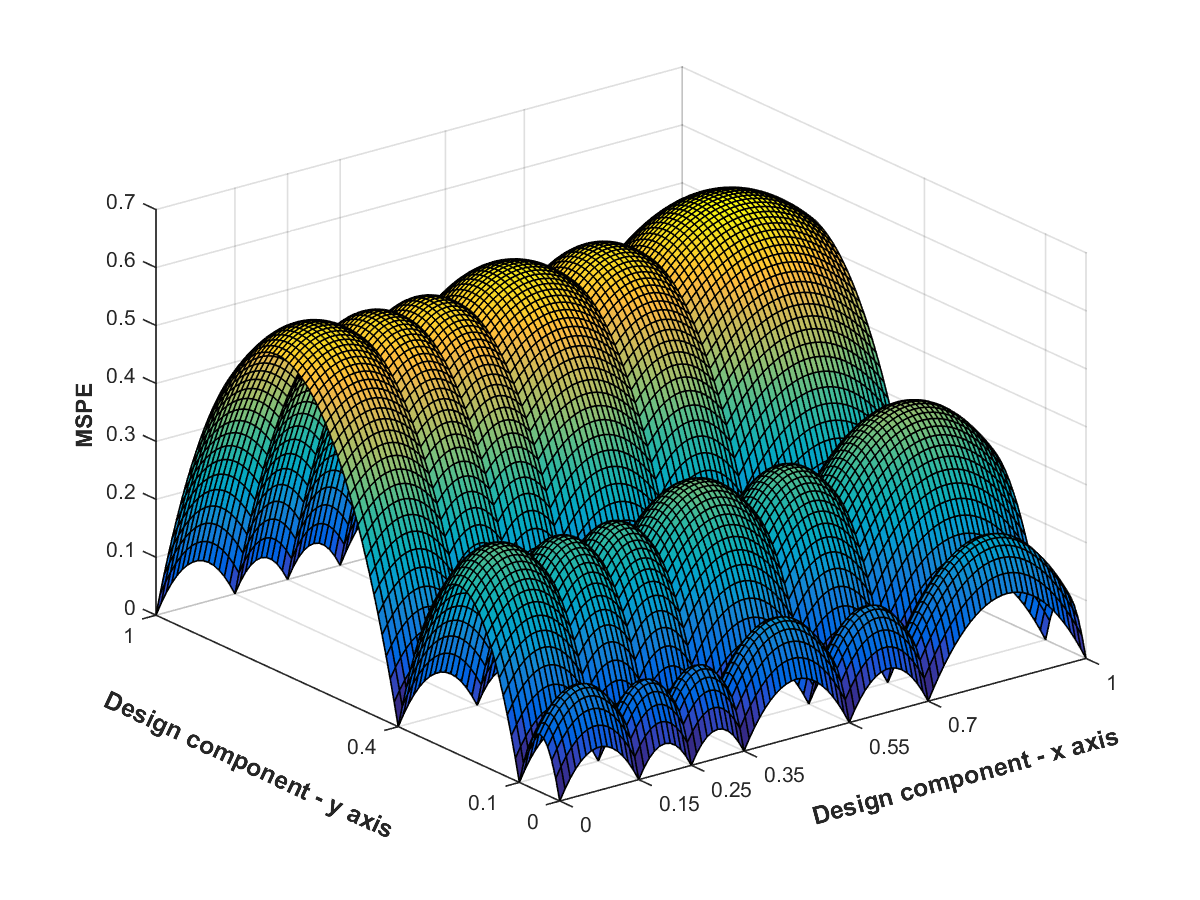}
			\captionsetup{justification=centering}
			\caption{\scriptsize $\mathcal{S}_{1} = \{0   , 0.15,    0.25,    0.35,     0.55,     0.7, 1\}$ and $\mathcal{S}_{2} =  \{0  ,  0.1,    0.4,    1\}$. $\sigma^{2} = 1$, $\alpha = 1$, and $\beta = 2$. }
		\label{fig:SMSPE OK II}
	\end{subfigure}
	\centering
	\caption{ Variation of MSPE for ordinary kriging with separable exponential covariance.}
	\label{fig:SMSPE OK} 
\end{figure}

The supremum of means squared prediction error, that is, the \textit{SMSPE} design criterion is given by,  
\begin{align}
SMSPE(\xibold, \CovParameterTheta) &= \sup_{(x_{0}, y_{0}) \in \mathcal{D}_{1} \times \mathcal{D}_{2}} MSPE((x_{0}, y_{0}), \xibold, \CovParameterTheta).  \label{SMSPE_II}
\end{align}

\begin{lemma} \label{Paper_II_LemmaSimpleAndOrdinary}
	Consider simple and ordinary kriging models, with isotropic random function $Z(\cdot)$ {and} separable exponential structure with parameters $\CovParameterTheta$. If $Z(\cdot)$ is observed over a region $[x_1, x_n] \times [y_1, y_m]$ and recorded at points on {a} grid $\xibold${, then the} MSPE at $(x_{0}, y_{0}) \in [x_{i}, x_{i+1}] \times [y_{j}, y_{j+1}]$  for some $ i = 1,\ldots,n-1$ and $ j = 1,\ldots,m-1$ is given by:
	\begin{align} 
	\sigma^{2}_{sk}(x_{0},y_{0}) 
	&= \sigma^{2} \Big( 1- \dfrac{ e^{-2\beta b } -2  e^{-2\beta d_{i}   } + e^{-2\beta (d_{i} - b ) }   }{ 1- e^{-2 \beta d_{i} }   } \;\; \dfrac{ e^{-2 \alpha a } -2  e^{-2\alpha d_{i}   } + e^{-2 \alpha (d_{i} - a ) }   }{ 1- e^{-2 \alpha d_{i} }   } \Big) \label{Paper_II_SimplekrigingMSPEIII} 
	\end{align}
	and
	\begin{align} 
	\sigma^{2}_{ok}(x_{0},y_{0}) 
	&= \sigma^{2} \Bigg( 1 - \dfrac{ e^{-2\beta b } -2  e^{-2\beta d_{i}   } + e^{-2\beta (d_{i} - b ) }   }{ 1- e^{-2 \beta d_{i} }   } \;\; \dfrac{ e^{-2 \alpha a } -2  e^{-2\alpha d_{i}   } + e^{-2 \alpha (d_{i} - a ) }   }{ 1- e^{-2 \alpha d_{i} }   } \nonumber \\
	& \;\;\;\;\;\;\;\;\;\;+ \dfrac{1}{\Omega_{x}(\xibold) \Omega_{y}(\xibold) }
	\Big(1-\dfrac{e^{- \beta b } + e^{- \beta (d_{i}  - b )  } }{ 1 + e^{- \beta d_{i}  }}  \;\; \dfrac{e^{- \alpha a } + e^{- \alpha(d_{i}  - a )  } }{ 1 + e^{- \alpha d_{i}  }} \Big)^2
	\Bigg), \label{Paper_II_OrdinarykrigingMSPEIII}
	\end{align}
	where $ a = x_{0} - x_{i}$, $ b = y_{0} - y_{j}$, $\Omega_{x}(\xibold) = \VectOneN^{T} \MatrixP^{-1} \VectOneN$, and $ \Omega_{y}(\xibold) = \VectOneM^{T} \MatrixQ^{-1} \VectOneM$. 	
\end{lemma}
\begin{proof}
{Substitute} the values {obtained in} \eqref{PaperII_decomposition1}, \eqref{PaperII_decomposition2}, \eqref{PaperII_decomposition3}, and \eqref{PaperII_decomposition4} from Appendix~\ref{PaperIIAppendixA} in equations~\eqref{Paper_II_SimplekrigingMSPEII} and \eqref{OrdinarykrigingMSPEII}. 

\end{proof}

\begin{remark} \label{Paper_II_Remark}
		Consider the simple and ordinary kriging models with the random function $Z(\cdot)$ as in Lemma~\ref{Paper_II_LemmaSimpleAndOrdinary}. {Recall that} $Z(\cdot)$ is observed over $[x_{1}, x_{n}] \times [y_{1}, y_{m}]$ and recorded at $\{x_{1}, \ldots, x_{n}\} \times \{y_{1}, \ldots ,y_{m}\}$. {Using the above expressions for} $MSPE_{sk}$ and $MSPE_{ok}$, we {obtain} that $Z(\cdot)$ is equivalent to an isotropic random process with exponential covariance parameters $(\sigma^{2}, (x_{n} - x_{1})  \alpha $ , $(y_{n}- y_{1}) \beta)$, observed over $[0, 1] \times [0, 1]$ and recorded at $\displaystyle{\{(x_{i}- x_{1})/(x_{n} - x_{1}): i = 1, \ldots, n \} \times \{(y_{j}- y_{1})/(y_{m} - y_{1}): j = 1, \ldots, m \}}$. 
\end{remark}
\begin{proof}
		The proof for simple kriging is given. Define a mapping $\chi_1(\cdot)$ over $[x_1, x_{n}]$ to $[0,1]$, such that for any point $x \in [x_{1}, x_{n}]$, $\chi_1(x) = (x-x_{1})/(x_{n}-x_{1})$ and $\chi_2(\cdot)$ over $[y_1, y_{m}]$ to $[0,1]$, such that for any point $y \in [y_{1}, y_{m}]$, $\chi_2(y) = (y-y_{1})/(y_{m}-y_{1})$. 
	
	If we take $g_{i} = d_{i}/(x_{n}-x_{1})$, $h_{i} = \delta_{j}/(y_{m}-y_{1})$ for $i = 1, \ldots, n-1$ and $j = 1, \ldots, m-1$, then the design $\xibold^\ast = ((0,g_{1},\ldots, g_{n-1}), (0,h_{1},\ldots, h_{m-1}))$ specifies the set of grid points $\{\chi_1(x_i): i = 1, \ldots, n\} \times \{\chi_2(y_j): j = 1, \ldots, m\} $. Here,  $\chi_1(x_1) = 0 $, $\chi_1(x_n)=1$, $\chi_2(y_1) = 0 $ and $\chi_2(y_m) = 1$. {Using} equation~\eqref{Paper_II_SimplekrigingMSPEIII}, it is very easy to see that:
\begin{align*}
MSPE_{sk}((x_{0}, y_{0}), \xibold, (\sigma^{2}, \alpha, \beta)) &= \\
 &  MSPE_{sk}((\chi_1(x_0) , \chi_2(y_0)), \xiboldAst, (\sigma^{2}, (x_{n} - x_{1}) \; \alpha, (y_{m} - y_{1}) \;\beta)).
\end{align*}		
As, $\chi_1(\cdot)$ and $\chi_2(\cdot)$ are bijective functions, we can assert our claim.

For ordinary kriging, {the proof is similar to the simple kriging case: use the functions $\chi_1(\cdot)$ and $\chi_2(\cdot)$ as above} and use equations \eqref{Paper_II_OrdinarykrigingMSPEIII}, \eqref{PaperII_Loc_Eq1}, and \eqref{PaperII_Loc_Eq2}. 
\end{proof} 

{Reasoning as in Remark~\ref{Paper_II_Remark},} any random process with constant mean, observed over $[x_{1}, x_{n}]$ could {be} viewed equivalently as a random process over $[0, 1] \times [0, 1]$ and {{\it vice versa}}. So, {without} loss of generality we take $\mathcal{D}_{1} \times \mathcal{D}_{2} = [0,1] \times [0,1]$ and $x_{1}=0$, $x_{n}=1$, $y_{1}=0$, and $y_{m}=1$. So, in the following sections we denote $\VectD = (d_{1},\ldots,d_{n-1}) $ and $\VectDelta = (\delta_{1},\ldots,\delta_{m-1})${,} omitting $x_{1} $ and  $y_{1}$. We will {frequently use the following notation in the rest of this paper:} $\displaystyle{ \norm{\VectD}_{\infty}  = \max_{i=1,\ldots,n-1} d_{i}}$ and $\displaystyle{\norm{\VectDelta}_{\infty} = \max_{j=1,\ldots,m-1} \delta_{j}}$.

\section{Prospective designs}	\label{Paper_II_Prospective_Design}
In this section, we find prospective optimal designs for simple and ordinary kriging, that is, we determine the design $\xibold$ which maximizes prediction accuracy before the onset of the statistical experiment. Both frequentist and Bayesian {paradigms} for the covariance parameters are considered. In the {case, where} the covariance parameters are treated as random variables, the optimal designs are found using a pseudo-Bayesian approach. Under both scenarios, an equispaced grid design is an optimal design (with respect to the $SMSPE$ criterion). 

\begin{theorem} \label{Paper_II_Theorem_ST_SK_SMSPE}  \label{Paper_II_Theorem_ST_OK_SMSPE} \label{Paper_II_Theorem_ST_SK_OK_SMSPE}
Consider the simple and ordinary kriging model with response $Z(\cdot)$ as in Lemma~\ref{Paper_II_LemmaSimpleAndOrdinary} and observed over $[0, 1] \times [0, 1]$. {An} \text{equispaced} grid design in both the input variables minimizes $SMSPE$. Thus, the equispaced grid design is the G-optimal design.
\end{theorem}
\begin{proof}
See Appendix  \ref{AppendixGMinus1}. 
\end{proof}

Note, the equispaced grid design minimize the $SMSPE$ for all covariance parameter values. Next, we assume that $\CovParameterTheta$ is a random variable with a known distribution. We use the pseudo-Bayesian approach to find the G-optimal design. The design criterion that is used here is: 
\begin{align}
	\mathcal{R}_{sk}(\xibold) &= E_{\CovParameterTheta}[SMSPE_{sk}(\xibold, \CovParameterTheta )] \text{ and } \\
\mathcal{R}_{ok}(\xibold) &= E_{\CovParameterTheta}[SMSPE_{ok}(\xibold, \CovParameterTheta)].
\end{align}  
	\begin{theorem} \label{Paper_II_Theorem_ST_SK_OK_SMSPE_PseudoBAyesian} 
	Consider the simple and ordinary kriging models with response $Z(\cdot)$ as in Lemma~\ref{Paper_II_LemmaSimpleAndOrdinary}, observed over $[0, 1] \times [0, 1]$, and the covariance parameters $\sigma^{2}$, $\alpha$, and $\beta$ are assumed to be independent with prior probability density functions $r_{1}(\cdot)$, $r_{2}(\cdot)${, and} $r_{3}(\cdot)$, respectively{, each} with bounded support. Then, an equispaced grid design minimizes both $\mathcal{R}_{sk}(\xibold)$ and $\mathcal{R}_{ok}(\xibold)$. 
	\end{theorem}
	\begin{proof}
	See Appendix  \ref{AppendixGMinus1}.  
	\end{proof}

 Theorems~\ref{Paper_II_Theorem_ST_SK_SMSPE} and \ref{Paper_II_Theorem_ST_SK_OK_SMSPE_PseudoBAyesian} suggest an equally spaced grid design would lead to the best prediction accuracy. However, in many practical scenarios for an ongoing statistical experiment, the design already exists. In such cases, there may be a need to improve the design {by} adding or deleting sampling points. In the next section, methods for finding the best retrospective design by minimizing the \textit{SMSPE} criteria are given.

\section{Retrospective design} \label{Paper_II_Retrospective_Design}
In this section, we provide algorithms for determining {the} best possible retrospective designs, with respect to the $SMSPE$ criterion for simple and ordinary kriging models. For finding the retrospective designs first the initial choice set containing the best possible design is constructed in such a way that it has finite cardinality (see Lemma~\ref{Paper_II_Retrospective_Lemma_I}). This ensures that the algorithms provided in this article are deterministic and converging. After the finite initial choice set is determined for the case of addition or deletion, we need to find the best possible design. One of the trivial ways is to compute the value of $SMSPE$ corresponding to each design in the choice set and then choose the one which gives the lowest value of $SMSPE$. However, this method has two major drawbacks: i) the value of covariance parameter is always needed and ii) if the value of the covariance parameters change, $SMSPE$'s needs to be computed again to determine the best design. To address these two challenges and avoid recalculating $SMSPE$ values for determining the best possible design a criterion for comparing evenness of designs is proposed in Definition \ref{Paper_II_Evenness_Defini}. Definition \ref{Paper_II_Evenness_Defini} gives a method to compare the evenness of two-dimensional grids, equivalently rank the grids in terms of their evenness. Then it is proved that under this criterion, indeed a more evenly spread grid leads to a lower value of $SMSPE$.

	In Section~\ref{Paper_II_Prospective_Design}, we saw that the optimal prospective designs for simple and ordinary kriging with two-dimensional inputs are equispaced, under a range of conditions. So, it is intuitive that retrospective designs should also be  as evenly spaced as possible. The notion of `evenly spread out' or `majorization' as given in \cite{MarshalOlkinBook} is used for comparing designs. Take two vectors ${\pmb w} =(w_{1}, w_{2}, \ldots, w_{r}) $ and ${\pmb w^{\prime}} =(w^{\prime}_{1}, w^{\prime}_{2}, \ldots, w^{\prime}_{r}) $ such that $ {\pmb w}, {\pmb w^{\prime}} \in \mathbb{R}^{r}$. {Suppose} $\displaystyle{w_{[1]} \geq w_{[2]} \geq \ldots  \geq w_{[r]}}$ denote the $r$ ordered components of the vector ${\pmb w}$. Then  ${\pmb w}$ is majorized ($\prec $) by ${\pmb w^{\prime}}$, under the following condition:  
	\begin{align}
		{\pmb w} \prec {\pmb w^{\prime}} & \text{ if } \begin{cases}
		&\sum_{i=1}^{k} w_{[i]} \leq \sum_{i=1}^{k} w^{\prime}_{[i]}, \;\;\; k=1, \ldots, r-1,\\
		& \sum_{i=1}^{r} w_{[i]} = \sum_{i=1}^{r} w^{\prime}_{[i]}.									 
		\end{cases} \label{Paper_II_Majorization_Definition}
		\end{align} 
		Under the partial order `$\prec$', ${\pmb w} $ is more evenly distributed than ${\pmb w^{\prime}}$. 
		\begin{definition}[Evenness of two-dimensional grids] \label{Paper_II_Evenness_Defini}
					Consider two grid designs, $\xibold $ and $\xibold^{\prime}$ of the same size (say $n \times m$). The designs are equivalently denoted $\xibold \equiv (\VectD, \VectDelta)$ and $\xibold^{\prime} \equiv (\VectD^{\prime}, \VectDelta^{\prime}) $. If $ \VectD \prec \VectD^{\prime}$ and $ \VectDelta \prec \VectDelta^{\prime}$, then we say $\xibold $ {is} more evenly spread than $\xibold^{\prime}$. 
		\end{definition}
In the following theorem we compare the vectors $ \VectD $, $\VectD^{\prime}$ and $ \VectDelta $, $\VectDelta^{\prime}$ in order to compare the corresponding designs with respect to the $SMSPE$ criteria.
	\begin{theorem} \label{Paper_II_Theorem_Any_Design}
	Consider a simple or ordinary kriging model with response $Z(\cdot)$ as in Theorem~\ref{Paper_II_Theorem_ST_SK_OK_SMSPE}. Consider two grid designs, $\xibold $ and $\xibold^{\prime}$ of the same size ($n \times m$). If, $ \VectD \prec \VectD^{\prime}$ and $ \VectDelta \prec \VectDelta^{\prime}$ then $SMSPE(\xibold) \leq SMSPE(\xibold^{\prime} )$ and therefore $\xibold $ is a better design than $\xibold^{\prime}$. 
\end{theorem}
\begin{proof} 
	See Appendix~\ref{AppendixG0} for the proof.

\end{proof}
	However, a direct application of Theorem~\ref{Paper_II_Theorem_Any_Design} in practice may not always be possible. Note, that `$\prec$' defines a partial order over a set. Hence, if there is a set of $n \times m $ grid designs $\SetUPrime$, for designs $\xibold, \xibold^\prime  \in \SetUPrime $, the respective vector pairs $\VectD$ and $ \VectD^{\prime}$, or $ \VectDelta $ and $\VectDelta^{\prime}$, might not be comparable under the partial order. In such a case, obtaining a unique design $ \xibold^{o} \equiv (\VectD^o ,\VectDelta^o) \in \SetUPrime$ (using Theorem~\ref{Paper_II_Theorem_Any_Design}) such that for any $\xibold^{\prime} \in \SetUPrime $, the vectors $\VectD^o \prec	 \VectD^{\prime}$ and $ \VectDelta^o \prec	 \VectDelta^{\prime}$ might not even be possible. Nevertheless, we will see that Theorem~\ref{Paper_II_Theorem_Any_Design} is very helpful for determining the best possible retrospective design and eliminating dependency of the best design on the values of covariance parameters in many cases.
	
		If the experimenter knows that the best possible retrospective design belongs to the set $\SetUPrime$, it would be useful if some designs could be eliminated from this set and then the experimenter could work with a much smaller subset of designs $ \SetU (\subseteq	\SetUPrime)$ which contains the best possible retrospective design. Lemma~\ref{Paper_II_Minorised_Designs} provides a method to find such a set $\SetU$. Note that this lemma is applicable to the cases where the initial choice set $(\SetUPrime)$ contains a finite number of designs. 
		\begin{lemma} \label{Paper_II_Minorised_Designs}
			Suppose that, for the purpose of conducting simple or ordinary kriging experiments as in Theorem~\ref{Paper_II_Theorem_ST_SK_OK_SMSPE}, we have an initial choice set of grid designs (all of size $n \times m $) denoted $\; \SetUPrime = \{ \xibold^\prime_{i}\equiv (\VectD^\prime_{i}, \VectDelta^\prime_{i} ) : i = 1, \ldots, \aleph\}$. {Suppose a subset $\SetU \subseteq \SetUPrime$ is constructed using the following algorithm:}
			\begin{algorithmic}[1]
				\STATE	Set, $ \SetU = \emptyset$. 
				\FOR{ $i = 1, \ldots, 	\aleph $}
				\IF{ for $\xibold^\prime_{i} \in \SetUPrime $, we cannot find any design $\xibold  \in \SetUPrime $ such that $\VectD \prec \VectD^{\prime}_{i}$ and $ \VectDelta \prec \VectDelta^{\prime}_{i}$ }
				\STATE{$\SetU = \SetU \cup \{\xibold^\prime_{i} \} $}.
				\ENDIF
				\ENDFOR
			\end{algorithmic}
			Then the set $\SetU $ contains the best possible design. 
		\end{lemma}
\begin{proof}
	From Theorem~\ref{Paper_II_Theorem_Any_Design} it is clear that if for $ \xibold^\prime_{i} \in \SetUPrime $ for some $i$, there exist $\xibold  \in \SetUPrime $ such that $\VectD \prec \VectD^{\prime}_{i}$ and $ \VectDelta \prec \VectDelta^{\prime}_{i}$ then, $SMSPE(\xibold ) \leq SMSPE(\xibold^\prime_{i})$. In that case $\xibold^\prime_{i}$ cannot be the best possible design and therefore {the} set $\SetU $ contains the best possible design. 
\end{proof}	
\noindent NOTE: Lemma~\ref{Paper_II_Minorised_Designs} is applicable for any choice of parameter values; in Step 3 of the lemma, we are comparing only the partition sizes of the design grids in order to eliminate some of the grids from $\SetUPrime$.

If $|\SetU| \ll  |\SetUPrime| $ or  $|\SetU| = 1$ then it would be easier to determine the best possible designs by comparing the designs in $\SetU$ with respect to the $SMSPE$ values. In Section~\ref{Paper_II_Illustration}, we see that in many cases $|\SetU| = 1$. Clearly, in such cases the selection of the best possible design {does} not depend upon the choice of parameter values. However, for cases where $|\SetU| \ll  |\SetUPrime|$ the best possible design depends on the choice of parameters. In that case, the advantage of using Lemma~\ref{Paper_II_Minorised_Designs} is that the experimenter may need to look at {very few designs} instead of all designs in the set $\SetUPrime$. 

Similarly, if we want to find the worst possible grid designs among the set $\SetUPrime$, Lemma~\ref{Paper_II_Majorised_Designs} provides a method to find {a} smaller set $\SetUWorst$ which contains the worst possible design.

\begin{lemma} \label{Paper_II_Majorised_Designs}
	Suppose that, for the purpose of conducting simple or ordinary kriging experiments as in Theorem~\ref{Paper_II_Theorem_ST_SK_OK_SMSPE}, we have an initial choice set of grid designs (all of size $n \times m $) denoted as $\; \SetUPrime = \{ \xibold^\prime_{i}\equiv (\VectD^\prime_{i}, \VectDelta^\prime_{i} ) : i = 1, \ldots, \aleph\}$. {Suppose a subset $\SetU \subseteq \SetUPrime$ is constructed using the following algorithm:}
	\begin{algorithmic}[1]
		\STATE	Set, $ \SetUWorst = \emptyset$. 
		\FOR{$i = 1, \ldots, 	\aleph $}
		\IF{for $\xibold^\prime_{i} \in \SetUPrime $, we cannot find any design $\xibold  \in \SetUPrime $ such that, $\VectD^{\prime}_{i} \prec \VectD $ and $ \VectDelta^{\prime}_{i}  \prec \VectDelta $ }
		\STATE{$\SetUWorst = \SetUWorst \cup \{\xibold^\prime_{i} \} $}.
		\ENDIF
		\ENDFOR
	\end{algorithmic}
	Then the set $\SetUWorst$ contains the worst possible design. 
\end{lemma}
\begin{proof}
The proof is similar to that of Lemma~\ref{Paper_II_Minorised_Designs}.
\end{proof}

In subsequent sections, the intuition that adding points to an existing design should lead to more accurate predictions is mathematically justified. This is followed by discussing two methodologies for finding the retrospective design by the addition of points: i) adding one point at a time and ii) adding all the points simultaneously. After this, a methodology for deleting points from an existing design is also discussed. 

\subsection{Addition of sampling points to the existing design}\label{Paper_II_Addition}

In this section, the task considered is that of adding $n_{1} $ new points, say $\{x_{1}^{\prime}, \ldots, x_{n_1}^{\prime} \}$, to $\mathcal{S}_{1}$ and $m_{1}$ new points, say $\{y_{1}^{\prime}, \ldots, y_{m_1}^{\prime} \}$, to $\mathcal{S}_{2}$. Since $x_{i}^{\prime}$ and $y_{j}^{\prime}$ ($i=1,\ldots,n_{1}$ and $j=1,\ldots,m_{1}$) are $n_{1} + m_{1}$ continuous variables over $(0,1)$, therefore there are infinite choices for constructing a new retrospective design (by adding new points to the existing design $\xibold$). Lemma~\ref{Paper_II_Retrospective_Lemma_I} constructs the initial choice set which contains the best possible design, such that the set has finite cardinality. In the following sections, the notation $\xiboldPlus \equiv (\VectDPlus, \VectDeltaPlus)$ is used to denote a new design obtained by adding points to an existing design. The new grid corresponding to $\xiboldPlus$ has $(n+n_{1}) \times (m+m_{1})$ points. Note, adding $n_{1}$ points to {the} $x$-covariate and $m_{1}$ points to {the}  $y$-covariate for an existing $n \times m$ grid would give rise to a new grid containing $(n m_{1} + m n_{1} + n_{1} m_{1})$ {additional} points. 

Intuitively, the addition of design points should lead to better prediction. However, looking at it mathematically (in Appendix~\ref{AppendixG1}) will help in determining the locations for adding points to an existing design ($\xibold$), such that the new design ensures the best prediction outcomes (with respect to the $SMSPE$ criterion).  
For the case of simple kriging, from Result \ref{AppendixG1_Result1} of Appendix~\ref{AppendixG1}, it follows that $SMSPE_{sk}(\xiboldPlus) \leq SMSPE_{sk}(\xibold)$. For the simple kriging models, there might be no reduction in $SMSPE$ despite adding new points. As an example, consider a design having $n$ points for {the} $x$-covariate and $\displaystyle{d_{i} = \dfrac{1}{n-1}}$ for all $i = 1, \ldots, n-1$. If the total number of points added to the $x$-covariate is less than $n-1$, then $SMSPE_{sk}$ remains unchanged. In this case, an experimenter needs to add at least $(n-1)$ additional points to see a reduction in $SMSPE_{sk}$. Whereas, for the case of ordinary kriging, from Result \ref{AppendixG1_Result2} of Appendix~\ref{AppendixG1}, it follows that $SMSPE_{ok}(\xiboldPlus) <  SMSPE_{ok}(\xibold)$. Thus, with every subsequent addition of points over any covariate axis of the grid, the $SMSPE_{ok}$ reduces. So for ordinary kriging models, depending on the budget{, an experimenter would be justified in adding} as many sample points as possible to an existing design.

Next, two different ways of adding points to an existing design are discussed. In Section~\ref{PaperII_Section_Algorithm1}, adding the points sequentially (or one at a time) is considered, while in Section~\ref{PaperII_Section_Algorithm2} adding points altogether (or simultaneously) is considered. Both the algorithms are deterministic. Lemma~\ref{Paper_II_Retrospective_Lemma_I} given below is helpful in designing Algorithms~\ref{PaperII_Algorithm1} and \ref{PaperII_Algorithm2}, used to find best retrospective designs by adding {points} to existing {designs} using {these} two different methods.

\begin{lemma} \label{Paper_II_Retrospective_Lemma_I}
	For simple and ordinary kriging models as in Theorem~\ref{Paper_II_Theorem_ST_SK_OK_SMSPE}, if $n_{1}$ and $m_{1}$ new points are added between $(x_{i}, x_{i+1})$ and $(y_{j}, y_{j+1})$, respectively for some $i = 1, \ldots, n-1$ and {$j= 1, \ldots, m-1$, then} equally spacing these new points within the {intervals} $(x_{i}, x_{i+1})$ and $(y_{j}, y_{j+1})$ leads to maximum reduction in the $SMSPE$ (equivalently, minimum {possible} $SMSPE$). 
\end{lemma}
			\begin{proof}
		See Appendix \ref{AppendixH}.
		\end{proof}

		\subsubsection{Sequential addition of points/ Adding one point at each stage} \label{PaperII_Section_Algorithm1}
In this section, the case in which the experimenter adds only one point to one of the covariate axis at a time (at each stage) is discussed for simple and ordinary kriging models. The single point is placed in the existing design in such a way that the new design with one more point on the chosen covariate axis is the best possible design at that stage. 
At each step, since only one point is added to one covariate axis, there is only one continuous variable over $(0, 1)$. Therefore, this problem is relatively easier than adding all points at the same time. 

In practice, this {case} would arise when an experimenter does not have {prior} knowledge of how many points should be added to the design. Rather, the experimenter adds a point, experiments for {some time} and then again adds another point if she thinks that would be necessary. 
\allowdisplaybreaks
\begin{lemma}
	For simple and ordinary kriging models as in Theorem~\ref{Paper_II_Theorem_ST_SK_OK_SMSPE}, Algorithm~\ref{PaperII_Algorithm1} (on page \pageref{PaperII_Algorithm1}) sequentially adds $n_1$ and $m_1$ points to the $x$- and $y$-covariates, respectively, in such a way that the minimum possible \textit{SMSPE} is attained at each stage.
\end{lemma}
\begin{algorithm}[h] 
	\caption{Retrospective design: Adding one point at a time}
	\label{PaperII_Algorithm1}
	{The initial design is given by $\xibold$.\\}
	{Set $\xiboldPlus_{a1} = \xibold$ (where $\xiboldPlus_{a1} \equiv (\VectDPlus_{a1}, \VectDeltaPlus_{a1})$). \\}
	{Set $k=1$ and $l=1$. \\}
	{Add $n_{1}$ points to $x$-covariate. }
	\begin{algorithmic}[1]
		\WHILE{$n_{1}> 0$}
			\STATE{Construct design $\xiboldPlus_{a1}$ by adding a new point $x_{k}^{\prime}$ to the design $\xibold$. The new  point $x_{k}^{\prime}$ is chosen to be the midpoint of $[x_{i_{0}}, x_{i_{0} + 1} ]$, such that $\norm{\VectD}_\infty = d_{i_{0}}$, where $d_{i_{0}} =  x_{i_{0} + 1} - x_{i_{0}} $ for some $i_{0} = 1, \ldots, length(\VectD)$.  $\dagger$\\
			(That is, $\xiboldPlus_{a1}$ is obtained by choosing the new point $x_{1}^{\prime}$ to be the midpoint of the biggest interval in the $x$-axis corresponding to $\xibold$).
		}
		\STATE{$n_1 \gets n_1 - 1$ and $k \gets k+ 1$  }
		\STATE{$\xibold \gets \xiboldPlus_{a1}$}
		
		\ENDWHILE 
	\end{algorithmic}
	
	Add $m_{1}$ points to $y$-covariate.
	\begin{algorithmic}[1]
		\WHILE{$m_{1}> 0$}
		\STATE{Construct design $\xiboldPlus_{a1}$ by adding a new point $y_{l}^{\prime}$ to the design $\xibold$. The new  point $y_{l}^{\prime}$ is chosen to be the midpoint of $[y_{j_{0}}, y_{j_{0} + 1} ]$, such that $\norm{\VectDelta}_\infty= \delta_{j_{0}}$, where $\delta_{j_{0}} =  y_{j_{0} + 1} - y_{j_{0}} $ for some $j_{0} = 1, \ldots, length(\VectDelta)$.  $\dagger$\\
			(That is, $\xiboldPlus_{a1}$ is obtained by choosing the new point $y_{1}^{\prime}$ to be the midpoint of the biggest interval in the $y$-axis corresponding to $\xibold$).
		}
		\STATE{$m_1 \gets m_1 - 1$ and $l \gets l+ 1$ }
		\STATE{$\xibold \gets \xiboldPlus_{a1}$}
		
		\ENDWHILE
	\end{algorithmic}

\end{algorithm}
	
\begin{proof} 
Since the terms in {the expressions for} $SMSPE_{sk}$ and $SMSPE_{ok}$ are separable for {the $x$- and $y$-covariates}, the order of addition of points (addition to $x$- followed by addition to $y$-covariate and {\it vice versa}) is not important. Step 2 of Algorithm~\ref{PaperII_Algorithm1} gives the best possible design at that stage; Appendix~\ref{AppendixG2} gives the proof. 
\end{proof}

	\subsubsection{Simultaneous addition of points/ Adding points altogether  } \label{PaperII_Section_Algorithm2}
		In this section, the problem of inserting $n_{1} + m_{1}$ new points simultaneously is considered. As noted earlier, this actually adds a larger number of points to the grid. Algorithms based on theory and computations are used to find the best possible retrospective design. In this case there are $(n_{1} + m_{1})$ new continuous variables over $(0,1)$. Since the choice set containing the best possible design is infinite, ways to narrow down the choices for the best possible design are needed. In fact, this step of shrinking the choice set ensures the convergence of Algorithm~\ref{PaperII_Algorithm2}.
		
\newpage		
	
\noindent	$^\dagger$ NOTE 1: $ d_{i_{0}}$ or $ \delta_{j_{0}}$ may not be unique and in case of tie choose one of the equal largest intervals arbitrarily.\\
\noindent NOTE 2:  The notation `${a1}$' in $\xiboldPlus_{a1}$ signifies that the design is obtained by Algorithm~\ref{PaperII_Algorithm1}.

To enumerate the ways to construct new designs by inserting $n_{1}$ new points to the $x$-covariate, note that each non-negative integer solution of the equation $\displaystyle{ n_1^{(1)} + n_1^{(2)} + \ldots + n_1^{(n-1)} = n_1}$ gives a way to distribute these $n_1$ points between existing design points over {the} $x$-covariate. 
			
		Denote the solution set of this equation by $\mathcal{T}_x$, then $\displaystyle{|\mathcal{T}_x| = \Comb{n_1 + n-2 }{n-2}}$. {Consider} the $k^{th}$ solution of the equation, $\displaystyle{(n_{1k}^{(1)}, n_{1k}^{(2)}, \ldots n_{1k}^{(n-1)}) }${.} 
		This corresponds to new designs constructed by inserting $n_{1k}^{(i)}$
		points between $(x_{i}, x_{i+1})$ for $i =1,\ldots,n-1$. {By Lemma~\ref{Paper_II_Retrospective_Lemma_I},} inserting these $n_{1k}^{(i)}$ points {equally spaced}  between $(x_{i}, x_{i+1})$ {minimizes the $SMSPE$ over this interval.} So, each non-negative integer solution of the equation $\displaystyle{ n_1^{(1)} + n_1^{(2)} + \ldots + n_1^{(n-1)} = n_1}$ corresponds to exactly one design which gives minimum value of the $SMSPE$ for that distribution of $n_1$ points. These finial designs minimizing the \textit{SMSPE} are unique {up to} the partition size. Similarly, {if $m_{1}$ new points are to be added to} the $y$-covariate, the ways these  $m_{1}$ points could be distributed can be obtained from the non-negative integer solution set of the equation $ \displaystyle{m_1^{(1)} + m_1^{(2)} + \ldots + m_1^{(m-1)} = m_1}$, denoted by $\mathcal{T}_y${, with} $\displaystyle{|\mathcal{T}_y| = \Comb{m_1 + m-2 }{m-2}}$. Here $\displaystyle{ (m_{1l}^{(1)} , m_{1l}^{(2)} , \ldots , m_{1l}^{(m-1)})}$ is the $l^{th}$ element of $\mathcal{T}_y$. Using Lemma\ref{Paper_II_Retrospective_Lemma_I} we know that once we know the number of points to be placed between a partition, the points are spaced equally to get best results. Therefore, the search for the best possible retrospective design is narrowed within $\displaystyle{\Comb{n_1 + n-2 }{n-2} \times \Comb{m_1 + m-2 }{m-2}}$ designs. Algorithm~\ref{PaperIIAlgorithm_AdditionChoiceSet} provides the initial choice set, say $\SetUTwoPrime$, for the best possible retrospective design obtained by simulataneous addition of points. 

\begin{algorithm}[]
	\begin{algorithmic}[1]
		\caption{Initial choice set construction - Simultaneous addition of points }
		\label{PaperIIAlgorithm_AdditionChoiceSet}
		\STATE	Set, $ \SetUTwoPrime = \emptyset$. \\
		\FOR{ $k = 1, \ldots, 	\Comb{n_1 + n-2 }{n-2} $ and $l = 1, \ldots, \Comb{m_1 + m-2 }{m-2}$}
		\FOR{$i = 1, \ldots, n-1 $ and $j = 1, \ldots, m-1 $}
		\STATE{Construct the design $\xibold^{+}_{kl} \equiv ({\pmb d^{+}_{k}},{\pmb \delta^{+}_{l}})$ by equally spacing $n_{1k}^{(i)}$ number of points between $(x_{i}, x_{i+1})$ and equally spacing $m_{1l}^{(j)}$ between $(y_{j}, y_{j+1})$ 
		}		
		\ENDFOR
		\STATE{$\SetUTwoPrime = \SetUTwoPrime \cup \{ \xibold^{+}_{kl} \}$}.
		\ENDFOR
	\end{algorithmic}
\end{algorithm}

	To choose the best retrospective design, the experimenter may compare values of $SMSPE$ over all the possible designs $\xibold^{+}_{kl} \in \SetUTwoPrime$ and determine which design minimizes the $SMSPE$. This method requires the values of covariance parameter $\CovParameterTheta$ and is sensitive to change in values of $\CovParameterTheta$. Algorithm~\ref{PaperII_Algorithm2} provides an approach that avoids the dependency of the best possible design on covariance parameters, for many cases. 
	
	\begin{lemma} \label{Paper_II_Lemma_algo2}
		For simple and ordinary kriging models as in Theorem~\ref{Paper_II_Theorem_ST_SK_OK_SMSPE}, Algorithm~\ref{PaperII_Algorithm2} gives the best possible retrospective grid design by simultaneously adding $n_{1}$ and $m_{1}$ points to {the} $x$ and $y$-covariate,  respectively.
	\end{lemma} 
\begin{algorithm}[h!]
	\caption{Retrospective designs: Adding all points {simultaneously}}
	\label{PaperII_Algorithm2}
	\begin{algorithmic}[1]
		\STATE{Find the choice set for designs $\SetUThreePrime$ (using Algorithm~ \ref{PaperIIAlgorithm_AdditionChoiceSet}).}
		\STATE{Use Lemma~\ref{Paper_II_Minorised_Designs} to find the subset $\SetUThree (\subseteq \SetUThreePrime) $ which contains the best possible design.}
		\IF{$|\SetUThree| = 1$}
		\STATE   $ \xiboldPlus_{a3} \equiv ({\VectDPlus_{a3}},{\VectDeltaPlus_{a3}}) \in \SetUThree $ is the only choice for best possible design. 
		\ELSE
		\STATE 		Set $\xiboldPlus_{a3} =  \underset{\xiboldPlus_{kl} \in \SetUThree}{\mathrm{arg \; min}} \; SMSPE(\xiboldPlus_{kl}, \CovParameterTheta) $\\
		
		\ENDIF
	\end{algorithmic}
\end{algorithm}

\begin{proof}
It is already discussed that $\SetUThreePrime$ contains the best possible design. It is clear, if Step 2 of Algorithm~\ref{PaperII_Algorithm2} determines a unique design, that is, if $|\SetUThree| = 1$, then, $\xiboldPlus_{a3} \in \SetUThree$ is the best possible retrospective design for any value of covariance parameter $\CovParameterTheta$. 		
\end{proof}

In upcoming sections, we will see that for many cases $|\SetUThree| = 1$ and in those cases Algorithm~\ref{PaperII_Algorithm2} is clearly advantageous over comparing $SMSPE$ values for all possible designs. In such cases, it is not necessary to calculate $SMSPE$ for each design, and thus the solution does not depend on the values of the covariance parameters.   
	\begin{theorem} \label{Paper-II_Corollary_I} \label{Paper_II_Theorem_Loc1}
	Adding points to a grid {simultaneously using} Algorithm~\ref{PaperII_Algorithm2} will never give a worse result than adding the points sequentially by Algorithm~\ref{PaperII_Algorithm1}.
	\end{theorem}
		\begin{proof}
		See Appendix~\ref{Paper_II_AppendixI}. 
		\end{proof}
		
		\noindent 	NOTE: There are many cases (evident from the Illustration) where $\xiboldPlus_{a2}$ is better than $\xiboldPlus_{a1}$.

\begin{remark} \label{Paper-II_Corollary_II}
	As a consequence of Theorems~\ref{Paper_II_Theorem_ST_SK_OK_SMSPE_PseudoBAyesian} and \ref{Paper_II_Theorem_Any_Design}, if we consider covariance {parameters to} be independent random variables with known probability density functions, we may apply the pseudo-Bayesian approach and minimize $\mathcal{R}_{sk}(\xibold) $  and  $\mathcal{R}_{ok}(\xibold) $ as in Section~\ref{Paper_II_Prospective_Design}{, using} independent priors for the covariance parameters as in Theorem~\ref{Paper_II_Theorem_ST_SK_OK_SMSPE_PseudoBAyesian}. In {the} case of addition of points, the best possible design obtained for one step at a time addition is as given by Algorithm~\ref{PaperII_Algorithm1}. If we want to add the points {simultaneously}, the new best possible designs can be {determined as} in Algorithm~\ref{PaperII_Algorithm2}, {except that in Step 6 we minimize $\mathcal{R}_{sk}(\xibold) $ or $\mathcal{R}_{ok}(\xibold) $ in place of $SMSPE_{sk}$ or $SMSPE_{ok}$.} 	
\end{remark}

	\subsection{Deleting of sampling points from an existing design} \label{Paper_II_Deletion}

The second problem that is considered is deletion of $n_{1}^{\prime} $ and $m_{1}^{\prime} $ points from $\mathcal{S}_{1}$ and $\mathcal{S}_{2}$, respectively. In this case, {choosing $n_{1}^{\prime} + m_{1}^{\prime}$ points from $\xibold$ results in only a finite number of possible designs} (the end points are always fixed at 0 and 1, hence there are $\displaystyle{\Comb{n-2}{n_{1}^{\prime}} \times \Comb{m-2}{m_{1}^{\prime}}}$ choices). In the following sections, the notation {$\xiboldMinus \equiv (\VectDMinus, \VectDeltaMinus)$ is used to denote a new design obtained by deleting points from the existing design.} Denote the choice set of such designs by $\displaystyle{\SetUFourPrime = \Bigg\{ \xiboldMinus_{kl} \equiv ( \VectDMinus_{k},\VectDeltaMinus_{l}): k= 1,\ldots, \Comb{n-2}{n_1}  \text{ and }  l = 1,\ldots, \Comb{m-2}{m_1} } \Bigg\}$. The most intuitive method for deleting points (simultaneously) from the existing design is by comparing the $SMSPE$ for each design. However, in that case the choice of design depends upon the covariance parameter values. In this section, an algorithm similar to Algorithm~\ref{PaperII_Algorithm2} is proposed which aims at eliminating the dependency on the covariance parameters and reducing the cardinality of the choice set $\SetUFourPrime$. In this case as well, the proposed Algorithm~\ref{PaperII_Algorithm3} is deterministic and the convergence is ensured as $|\SetUFourPrime|$ is finite.

\begin{lemma} 
	For simple and ordinary kriging models as in Theorem~\ref{Paper_II_Theorem_ST_SK_OK_SMSPE}, Algorithm~\ref{PaperII_Algorithm3} (on page \pageref{PaperII_Algorithm3}) gives the best possible retrospective grid design by simultaneously deleting $n_{1}^\prime$ and $m_{1}^\prime$ points from {the} $x$-covariate and $y$-covariate, respectively.	
\end{lemma}
\begin{algorithm}[h!]
	\caption{Retrospective designs (deletion case): Best possible retrospective design}
	\label{PaperII_Algorithm3}
	\begin{algorithmic}[1]
		\STATE{Find the choice set for designs $\SetUFourPrime$ by taking all possible designs obtained after deleting $n_{1}^\prime$ and $m_{1}^\prime$ points from $x$-covariate and $y$-covariate, respectively.}
		\STATE{Use Lemma~\ref{Paper_II_Minorised_Designs} to find the subset $\SetUFour (\subseteq \SetUFourPrime) $ which contains the best possible design.}
		\IF{$|\SetUFour| = 1$}
		\STATE   $ \xiboldMinus_{a4} \equiv ({\VectDMinus_{a4}},{\VectDeltaMinus_{a4}}) \in \SetUFour $ is the only choice for best possible design 
		\ELSE
		\STATE 		Set $\xiboldMinus_{a4} =  \underset{\xiboldMinus_{kl} \in \SetUFour}{\mathrm{arg \; min}} \; SMSPE(\xiboldMinus_{kl}, \CovParameterTheta) $
		\ENDIF
	\end{algorithmic}
\end{algorithm}
\begin{proof}
See proof {of} Lemma~\ref{Paper_II_Lemma_algo2}.
\end{proof}
	
	Next, an approach similar to Algorithm~\ref{PaperII_Algorithm3} is proposed to find the worst possible design. Algorithm \ref{PaperII_Algorithm4} is used to compare the best and the worst possible retrospective designs in Section~\ref{Paper_II_Illustration}. 
	
\begin{lemma} 
	For simple and ordinary kriging models as in Theorem~\ref{Paper_II_Theorem_ST_SK_OK_SMSPE}, {Algorithm}~\ref{PaperII_Algorithm4} gives the worst possible retrospective grid design by simultaneously deleting $n_{1}^\prime$ and $m_{1}^\prime$ points from {the} $x$-covariate and $y$-covariate, respectively.	
\end{lemma}
\begin{algorithm}[H]
	\caption{Retrospective designs (deletion case): Worst possible retrospective design}
	\label{PaperII_Algorithm4}
	\begin{algorithmic}[1]
		\STATE{Find the choice set for designs $\SetUFourPrime$ by taking all possible designs obtained after deleting $n_{1}^\prime$ and $m_{1}^\prime$ points from $x$-covariate and $y$-covariate, respectively.}
		\STATE{Use Lemma~\ref{Paper_II_Majorised_Designs} to find the set $\SetUFiveWorst (\subseteq \SetUFourPrime) $ which contains the worst possible design.}
		\IF{$|\SetUFiveWorst| = 1$}
		\STATE   $ \xiboldMinus_{a5-worst} \equiv ({\VectDMinus_{a5-worst}},{\VectDeltaMinus_{a5-worst}}) \in \SetUFiveWorst $ is the only choice for worst possible design 
		\ELSE
		\STATE 	Set $\xiboldMinus_{a5-worst} =  \underset{\xiboldMinus_{kl} \in \SetUFiveWorst}{\mathrm{arg \; max}} \; SMSPE(\xiboldMinus_{kl}, \CovParameterTheta) $
		\ENDIF
	\end{algorithmic}
\end{algorithm}
\begin{proof}
See proof {of} Lemma \ref{Paper_II_Lemma_algo2}.
\end{proof}
NOTE: In Step 6 of Algorithms~\ref{PaperII_Algorithm2}, \ref{PaperII_Algorithm3}, and \ref{PaperII_Algorithm4} values of the covariance parameters are required. 

\begin{remark} \label{Paper-II_Corollary_III}
	As a consequence of Theorem~\ref{Paper_II_Theorem_Any_Design} and Algorithms~\ref{PaperII_Algorithm3} and \ref{PaperII_Algorithm4}, if we consider covariance {parameters to} be independent random variables with known probability density functions, we {may apply the pseudo-Bayesian} approach and work with $\mathcal{R}_{sk}(\xibold) $  and  $\mathcal{R}_{ok}(\xibold) $ by assuming independent priors for the covariance parameters. In case of simultaneous deletion of points from an existing grid, the best and the worst possible design could be determined {as in} Algorithms~\ref{PaperII_Algorithm3} and \ref{PaperII_Algorithm4}, respectively{, except that in Step 6 we minimize $\mathcal{R}_{sk}(\xibold) $  or  $\mathcal{R}_{ok}(\xibold) $ in place of $SMSPE_{sk}$ or $SMSPE_{ok}$.} 	
\end{remark}

	\section{Illustration} \label{Paper_II_Illustration}

	In this section, the proposed Algorithms~\ref{PaperII_Algorithm1}, \ref{PaperII_Algorithm2}, \ref{PaperII_Algorithm3}, and \ref{PaperII_Algorithm4} for an ordinary kriging model under a frequentist paradigm are illustrated. In Illustration~\ref{Paper_II_Ilustration_1}, the retrospective designs obtained by adding points to the axis of the grid using Algorithms~\ref{PaperII_Algorithm1} and \ref{PaperII_Algorithm2} are found. In Illustrations~\ref{Paper_II_Ilustration_2} and \ref{Paper_II_Ilustration_3}, retrospective designs are found by deleting points from an existing design using Algorithms~\ref{PaperII_Algorithm3} and \ref{PaperII_Algorithm4}. {Illustration}~\ref{Paper_II_Ilustration_3} in particular takes an $8 \times 8$ regular grid used for monitoring methane flux as in \cite{baran2015optimal} and shows how a smaller retrospective design could be obtained by deleting points. An efficiency criteria is defined below \citep[similar to][]{Dette_et_al_2008} for comparing the designs: 
	\begin{align}
	eff(\xibold_{2} : \xibold_{1} ) &= 	\dfrac{SMSPE(\xibold_{1})}{SMSPE(\xibold_{2})}, 
	\end{align}
	where $eff(\xibold_{2} : \xibold_{1} ) $ is the efficiency of the design $\xibold_{2}$ with respect to $\xibold_{1}$. The efficiencies of the new retrospective designs are calculated: i) with respect to the initial design and ii) with respect to the equispaced design {of the same size as} the retrospective design. 
	
	As before, denote the initial design by $\xibold$. Retrospective designs obtained by addition and deletion are denoted by $\xiboldPlus$ and $\xiboldMinus$, respectively. Let, $\xiboldPlus_{eq}$ and $\xiboldMinus_{eq}$ be the equispaced grid designs {of the same size} as that of $\xiboldPlus$ and $\xiboldMinus$, respectively. 
	Clearly, {the} higher the value of efficiency of a design, {the} better is the chosen retrospective design. Note, that $eff(\xiboldPlus: \xibold) \geq 1$, $eff(\xiboldMinus: \xibold ) \leq 1$, $eff(\xiboldPlus: \xiboldPlus_{eq} ) \leq 1 $, and $eff(\xiboldMinus: \xiboldMinus_{eq} ) \leq 1$.

	\begin{illustration} \label{Paper_II_Ilustration_1}
		Consider a random process $Z(\cdot)$ with constant but unknown mean and a separable exponential covariance structure. The samples are initially taken over a $4 \times 5$ grid $\xibold = ((.80, .10, .10),  (.20, .10, .10, .60))$. The aim is to determine the best possible grid design ({minimizing} $SMSPE$), by adding 3 points to the \textit{x}-covariate and 2 points to the \textit{y}-covariate.   
	\end{illustration}
	Here, an ordinary kriging setup is considered. The initial design is given by $\xibold $, which means the $x$ and $y$-covariates of the design are $\mathcal{S}_{1} = \{0    ,0.8, 0.9,  1\}$ and $ \mathcal{S}_{2} = \{ 0  ,  0.2, 0.3, 0.4, 1\}$, respectively. It is given that $n = 4, m = 5 $ and $n_{1} =3 , m_{1} = 2$, which means the final design is a $7\times 7$ grid. As this is a case of addition of points to $\xibold$, initially there are {infinitely many} choices for constructing the retrospective design. Algorithms~\ref{PaperII_Algorithm1} and \ref{PaperII_Algorithm2} are used to determine the best possible retrospective designs and denoted by $\xiboldPlus_{a1}$ and $\xiboldPlus_{a3}$, respectively. 
	
	It was discussed earlier, that Algorithm~\ref{PaperII_Algorithm1} does not depend on values of the covariance parameter $\CovParameterTheta$: it needs only the initial design to determine the final design. A unique retrospective design is obtained (unique up to the partition sizes of the design). 
	
	However, Algorithm~\ref{PaperII_Algorithm2} might require the values of the covariance parameter $\CovParameterTheta$ in order to compute and compare $SMSPE$ if $|\SetUThree| \neq 1$. For this example, Step 1 of Algorithm~\ref{PaperII_Algorithm2} gives that the {best} possible retrospective design belongs to the set $\SetUThreePrime$ such that $|\SetUThreePrime| = 100$. In Step 2 the size of the choice set is further reduced and $|\SetUThree| = 1$. So, in this case, the retrospective design does not depend on the {covariance} parameters ({and} there is no need to compute and compare the $SMSPE$ for designs in $\SetUThree$). 
	
	The equispaced design of size $7 \times 7$ is denoted $\xiboldPlus_{{eq}_{ 7\times 7}}$ and is used for calculating the efficiency values. The retrospective designs $\xiboldPlus_{a1}$ and $\xiboldPlus_{a3}$ are compared with $\xibold$ and $\xiboldPlus_{{eq}_{ 7\times 7}}$ and are shown in Figure~\ref{Fig_Ret_Illus1_All}. 
		\begin{figure}[H]
			\centering
			\begin{subfigure}{.5\textwidth}
			\centering
			\includegraphics[width=.9\linewidth]{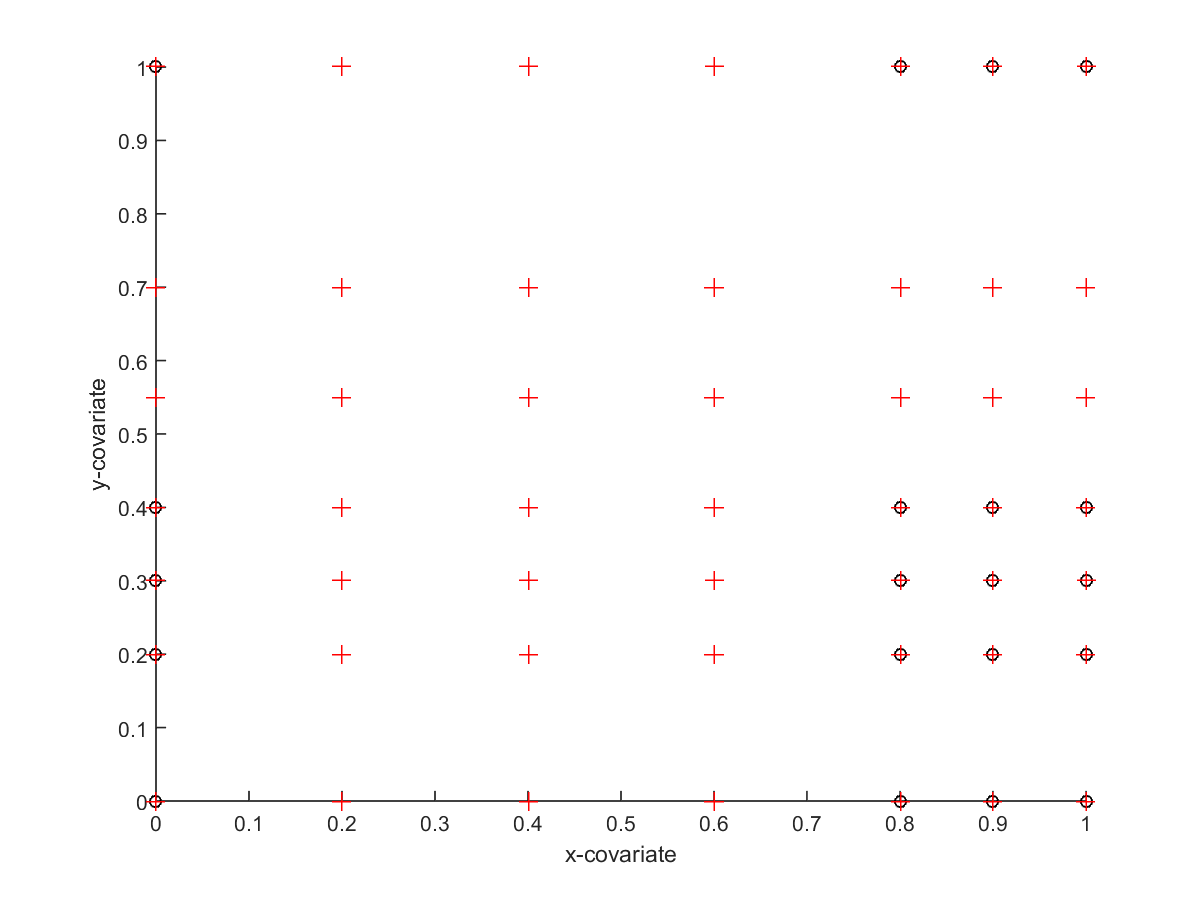}
		\caption{\scriptsize $\xiboldPlus_{a1}$ Vs $\xibold$    }
		\label{Fig_Ret_Illus1_1}
		\end{subfigure}%
		\begin{subfigure}{.5\textwidth}
		\centering
		\includegraphics[width=.9\linewidth]{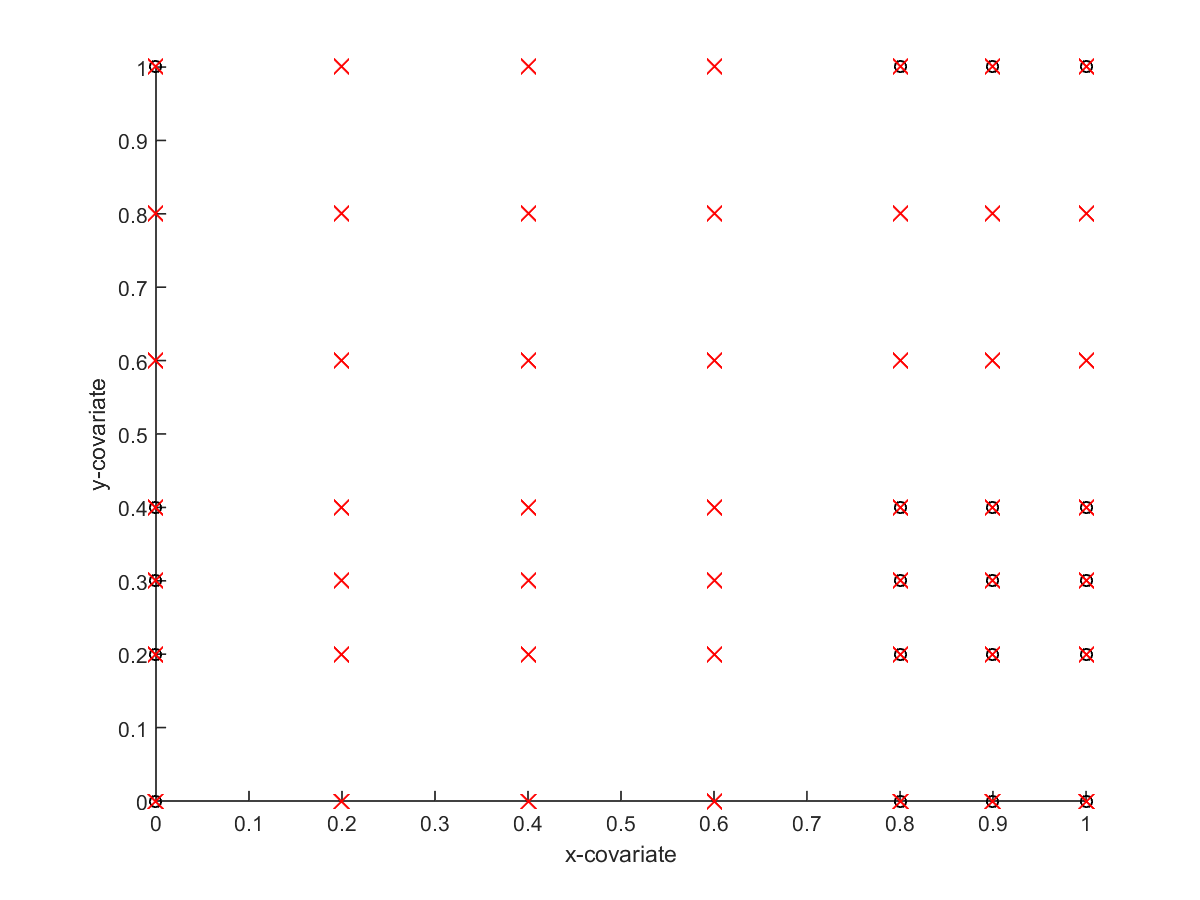}
		\caption{\scriptsize $\xiboldPlus_{a3}$  Vs $\xibold$   }
		\label{Fig_Ret_Illus1_2}
		\end{subfigure}
		\vspace{-1em}
		\centering
		\begin{subfigure}{.5\textwidth}
		\centering
		\includegraphics[width=.9\linewidth]{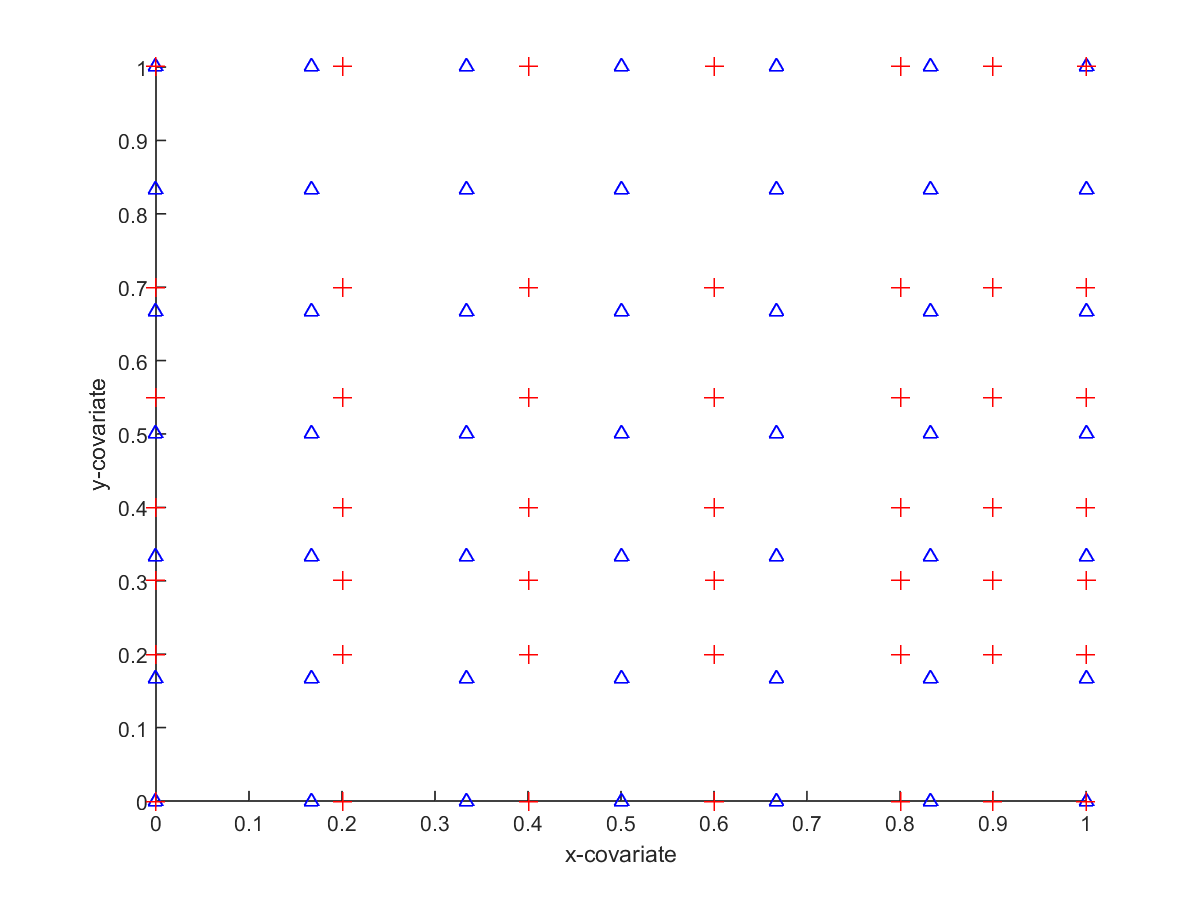}
		\caption{ \scriptsize $\xiboldPlus_{a1}$ Vs $\xiboldPlus_{{eq}_{ 7\times 7}}$  }
		\label{Fig_Ret_Illus1_3}
		\end{subfigure}%
		\begin{subfigure}{.5\textwidth}
			\centering
			\includegraphics[width=.9\linewidth]{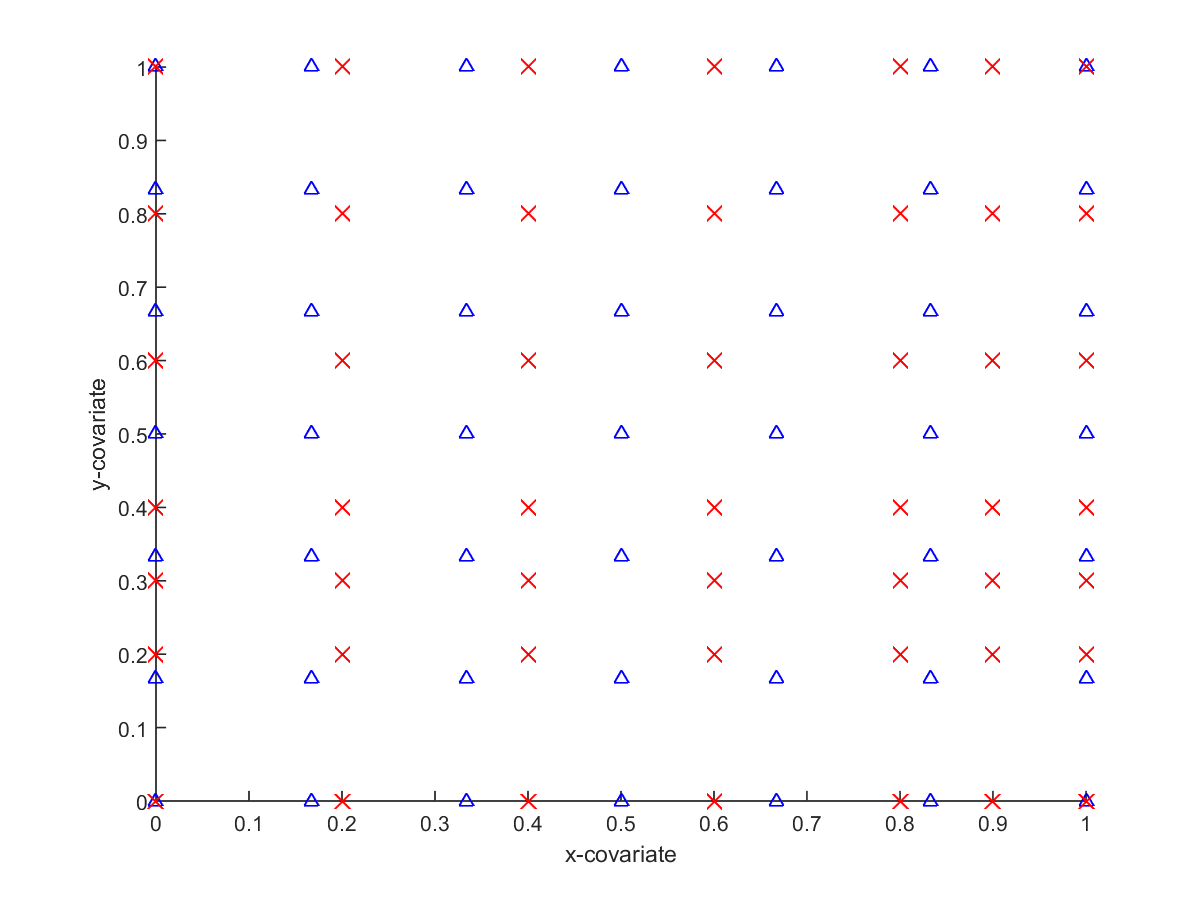}
		\caption{ \scriptsize $\xiboldPlus_{a3}$ Vs  $\xiboldPlus_{{eq}_{ 7\times 7}}$    }
		\label{Fig_Ret_Illus1_4}
	\end{subfigure}
	\vspace{1em}
	
	
	
	\centering	\caption{ Comparison of grid designs. 
	`$\color{red} \times$' - $\xiboldPlus_{a3}$: Best possible $7\times 7$ retrospective design obtained by Algorithm~3. 				 `$\color{red} +$' - $\xiboldPlus_{a1}$: Best possible $7\times 7$ retrospective design obtained by Algorithm~1.  
	`$\circ$' - $\xibold$: Original $4 \times 5$ design grid.  `$\color{blue}\triangle$' - $\xiboldPlus_{{eq}_{ 7\times 7}}$: Equispaced grid of size $7\times 7$.		
	}
	\label{Fig_Ret_Illus1_All} 
	\end{figure}
	In this example, Algorithms~\ref{PaperII_Algorithm1} and \ref{PaperII_Algorithm2} do not require values of the covariance parameter, since there is no need to compute and compare the vales of $SMSPE$ for both the algorithms. 
	
	However, for the purpose of illustration, three sets of values of covariance parameter are used to compare efficiencies of the newly obtained retrospective designs (with respect to the existing design and the prospective optimal design), as to get the values of efficiencies the \textit{SMSPE} needs to be calculated, which depends on the parameter value. The results are shown in Table~\ref{Efficence_Ills_1}.  
	
	\begin{table}[H]
	\begin{center}
	{
	\begin{tabular}{|l|cc|  c c|}
		\hline
		$(\alpha, \beta)$
		& $eff(\xiboldPlus_{a1}: \xibold )$ 
		& $eff(\xiboldPlus_{a3}: \xibold )$ 
	& $eff(\xiboldPlus_{a1}: \xiboldPlus_{{eq}_{ 7\times 7}} )$
	& $eff(\xiboldPlus_{a3}: \xiboldPlus_{{eq}_{ 7\times 7}} )$\\
	\hline
	(.5, .7) & 2.4401 & 3.1324  & 0.6526 & 0.8378 \\
	(1, 5)   & 1.5070 & 1.9832  & 0.6535  & 0.8600  \\
	(10, 15) & 1.0597 & 1.0869  & 0.9442  & 0.9684  \\
	\hline
	\end{tabular}	
	}
	\end{center}
	\caption{Efficiencies of the new designs $\xiboldPlus_{a1}$ and $\xiboldPlus_{a3}$, with respect to $\xibold$ and $\xiboldPlus_{{eq}_{ 7\times 7}} $}.
		\label{Efficence_Ills_1}
		\end{table}
			
		As expected, in Table~\ref{Efficence_Ills_1}, see that $ \xiboldPlus_{{eq}_{ 7\times 7}}  $ performs the best for all parameter values. The efficiency of $\xiboldPlus_{a3}$ is higher than $\xiboldPlus_{a1}$ with respect to both $\xibold$ and $ \xiboldPlus_{{eq}_{ 7\times 7}}  $ as expected due to Theorem~\ref{Paper-II_Corollary_I}. The values of $eff(\cdot: \xibold)$ {suggest} that the new design leads to considerable reduction in $SMSPE$.

			\begin{illustration} \label{Paper_II_Ilustration_2}
			Same as in Illustration~\ref{Paper_II_Ilustration_1}, we consider an ordinary kriging setup, where the random process $Z(\cdot)$ has a separable exponential covariance structure. The samples are initially taken over a $17 \times 5$ grid. The initial grid $\xibold = (\VectD, \VectDelta)$ to be $\VectD = (
			0.0500,    0.0700,    0.0400,    0.0250,    0.0410,  \\  0.0644,    0.0854, 0.0290,    0.1050, 
			0.0291,    0.1299,    0.1074,   0.0798,    0.0340,    0.0341,    0.0759)$ and $\VectDelta = (     0.2281,   \\ 0.1219,
		0.1446,    0.5054)$. The aim is to find the best possible grid design by deleting 10 points from {the} $x$-covariate and 2 points from {the} ${y}$-covariate.   
		\end{illustration}
		The initial design is given by $\xibold $, which means the $x$ and $y$-covariates of the design are given by 
		$\mathcal{S}_{1} = \{0,  0.0500 ,   0.1200   , 0.1600  ,  0.1850    , 0.2260    , 0.2904,  0.3758    , 0.4047    , 0.5097    ,  0.5388,  \\   0.6687   , 0.7762    , 0.8560, 0.8900    , 0.9241   , 1\}$ and $\mathcal{S}_{2} = \{    0  ,  0.2281    , 0.3500    , 0.4946 ,    1\}$, respectively. It is given, that $n = 17, m = 5 $ and $n_{1}^\prime =10 , m_{1}^\prime = 2$, so the final required design is a $7\times 3$ grid. We use Algorithms~\ref{PaperII_Algorithm3} and \ref{PaperII_Algorithm4} to determine the best and the worst possible retrospective designs. The size of {the} initial choice set for choosing designs is 9009, that is $|\SetUFourPrime| = 9009$. 
		
		We run Algorithm~\ref{PaperII_Algorithm3} {to} obtain the best possible design. In Step 2 of the algorithm, we get $|\SetUFour| = 1$. So, there is only one choice for the best possible retrospective design. Also, the best possible design does not depend on the values of the covariance parameters (since there is no need to compute and compare the $SMSPE$ for designs in $\SetUFour$). We denote the best possible design by $\xiboldMinus_{a4}$.\\ 
			
			{To obtain} the worst possible design we run Algorithm~\ref{PaperII_Algorithm4}. After executing Step 2 of the algorithm we get $|\SetUFiveWorst| = 5$. So, we have to perform the computational part in the algorithm, that is, Step 3. We take some values of $(\alpha, \beta)$ and for each case find the worst possible design ({which} maximizes $SMSPE$). The parameter values taken for the sake of illustration are (.5, .7),  (1, 5), and (10, 15). For each set of parameter {values} we find that the worst possible design is given by $\xiboldMinus_{a5-worst} =
			((0.0500  ,\;  0.0700 ,\;    0.0400 ,\;   0.0250  ,\;   0.0410 ,\;   0.7740 ),\\\; (    0.2281,\;    0.7719 ))$. \\
			
		We denote the $7\times 3$ equispaced design by $\xiboldMinus_{{eq}_{ 7\times 3}}$. In Figure~\ref{Fig_Ret_Illus1_All_illust2}, the retrospective designs $\xiboldMinus_{a4}$ and $\xiboldMinus_{a5-worst}$ are compared with $\xibold$ and and $\xiboldMinus_{{eq}_{ 7\times 3}}$. 
		\begin{figure}[H]
			\centering
			\begin{subfigure}{.5\textwidth}
			\centering
			\includegraphics[width=.9\linewidth]{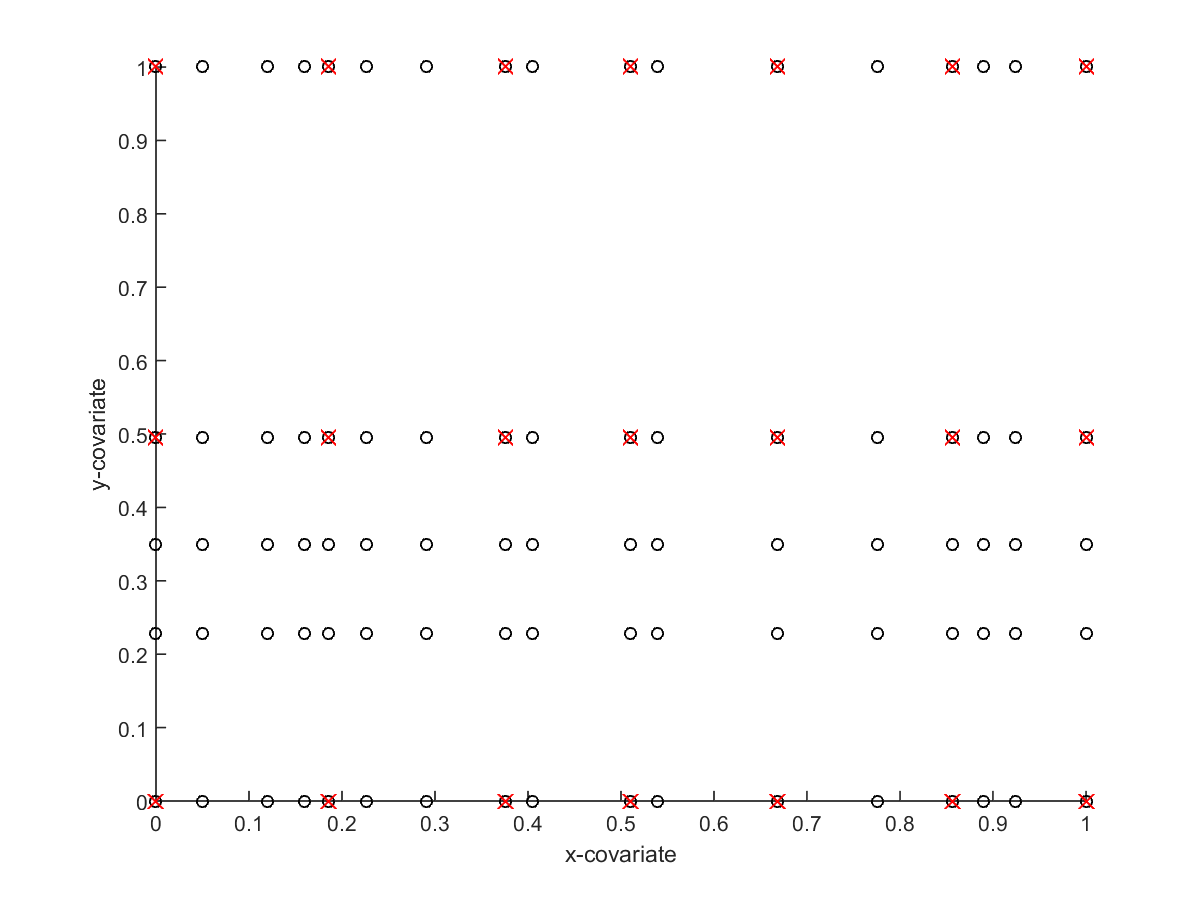}
		\caption{ \scriptsize $\xiboldMinus_{a4}$ Vs $\xibold$   }
		\label{fig:ImgIllust2_1}
		\end{subfigure}%
		\begin{subfigure}{.5\textwidth}
		\centering
			\includegraphics[width=.9\linewidth]{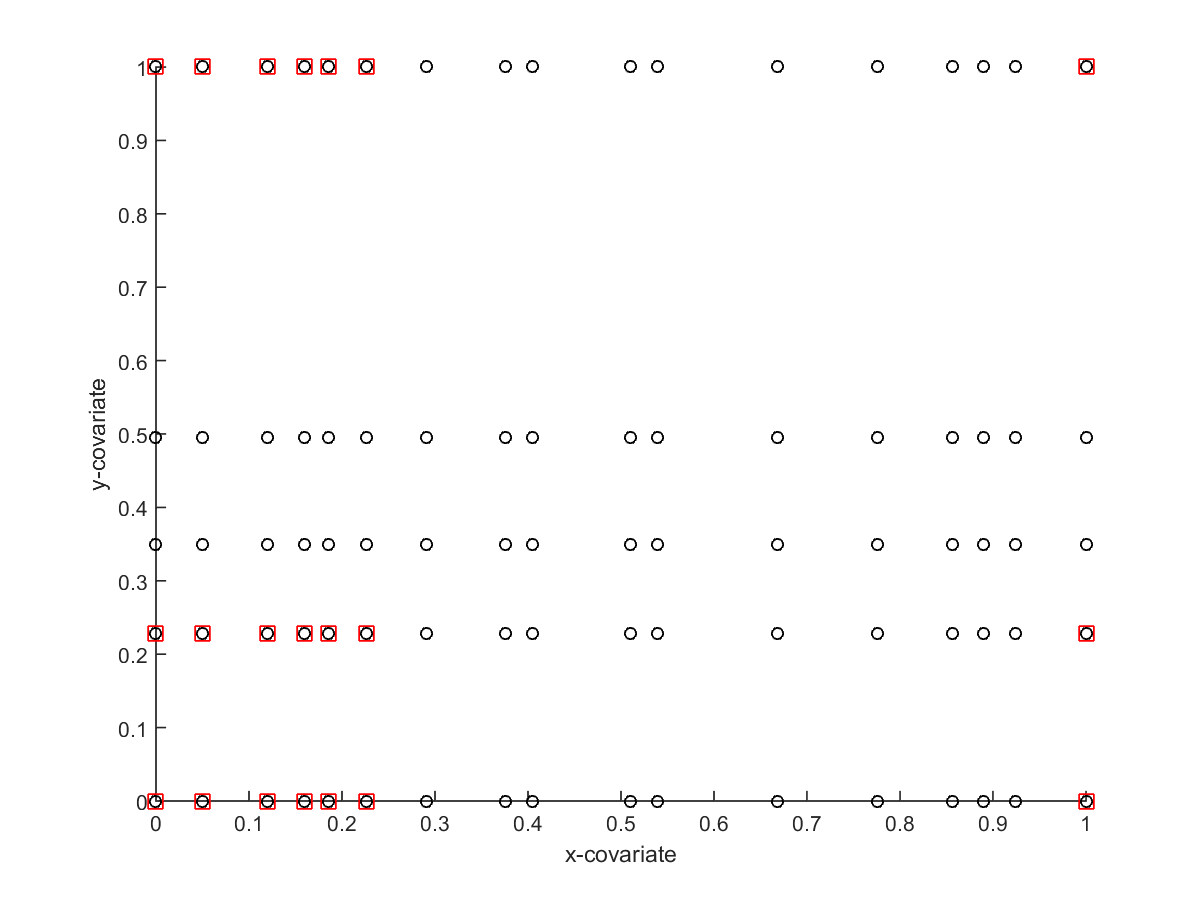}
			\caption{\scriptsize $\xiboldMinus_{a5-worst}$ Vs  $\xibold$   }
			\label{fig:ImgIllust2_2}
		\end{subfigure}
	\vspace{1em}
	\centering
	\begin{subfigure}{.5\textwidth}
		\centering
			\includegraphics[width=.9\linewidth]{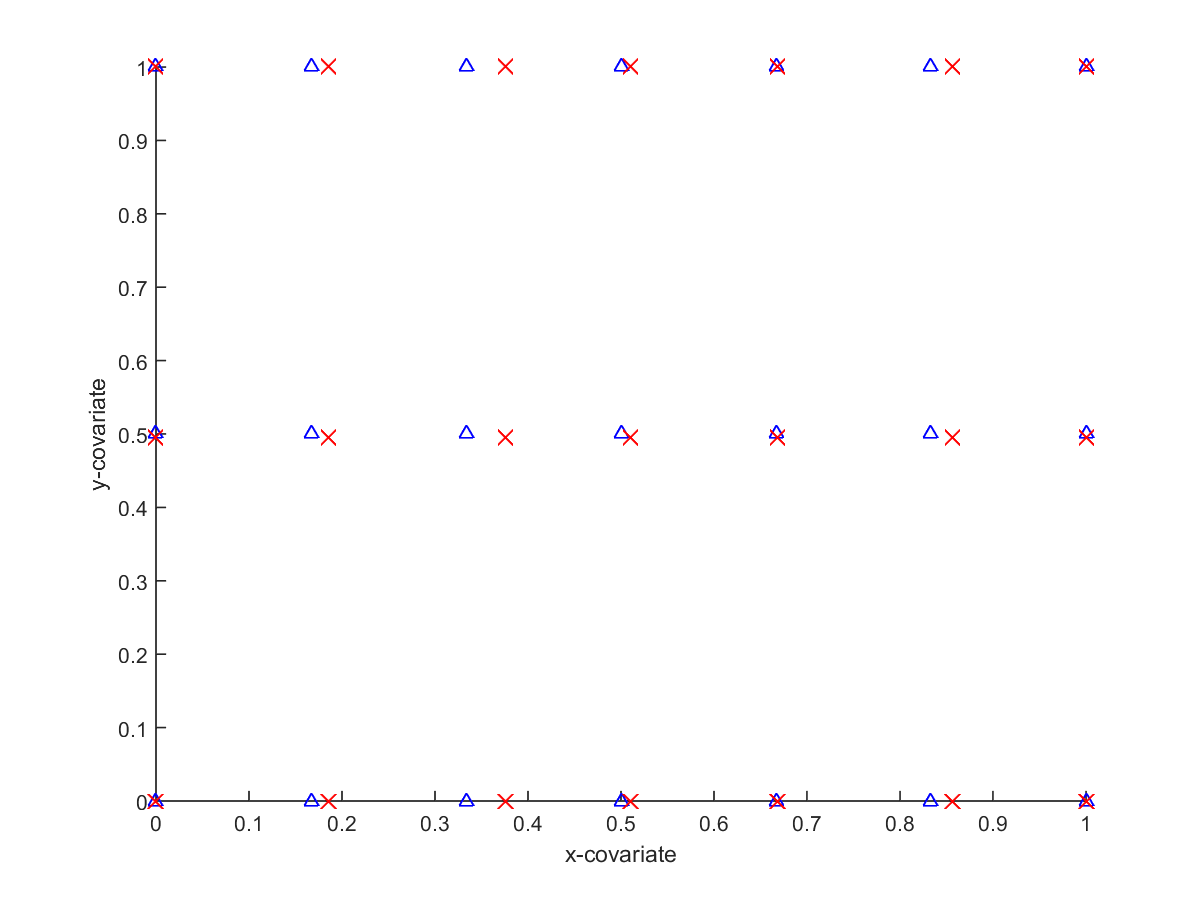}
				\caption{ \scriptsize $\xiboldMinus_{a4}$ Vs $\xiboldMinus_{{eq}_{ 7\times 3}}$  }
					\label{Fig_Ret_Illus2_3}
					\end{subfigure}%
					\begin{subfigure}{.5\textwidth}
					\centering
					\includegraphics[width=.9\linewidth]{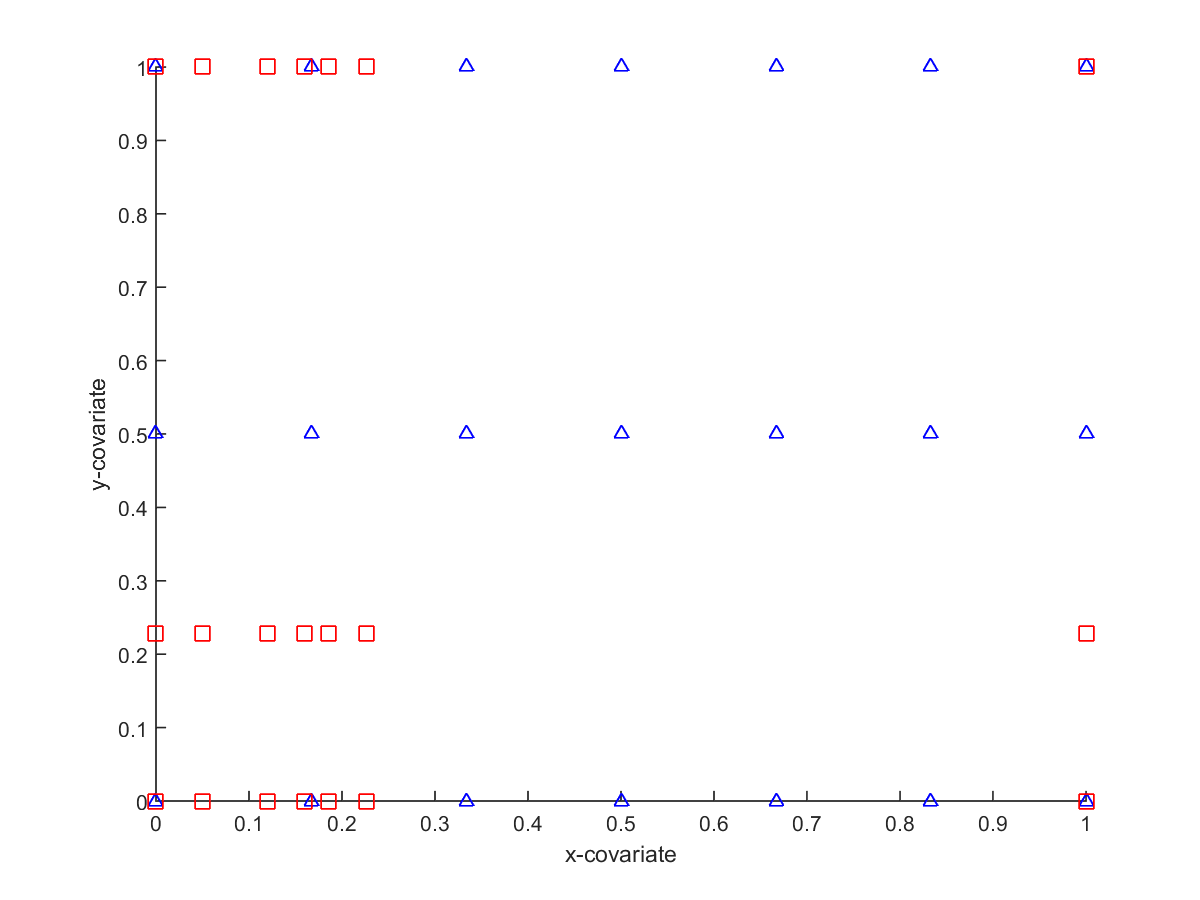}
					\caption{\scriptsize $\xiboldMinus_{a5-worst}$   Vs  $\xiboldMinus_{{eq}_{ 7\times 3}}$ }
					\label{Fig_Ret_Illus2_4}
					\end{subfigure}
					\vspace{-1em}
					\centering
				\caption{Comparison of design grids. 
			`$\color{red} \times$' - $\xiboldMinus_{a4}$: Best possible $7\times 3$ design. 	 		
		`$\color{red} \square$' - $\xiboldMinus_{a5-worst}$: Worst possible $7\times 3$ grid. 
		`$\circ$' - $\xibold$: Original $17 \times 5$ design grid. `$\color{blue}\triangle$' - $\xiboldMinus_{{eq}_{ 7\times 3}}$: Equispaced $7\times 3$ grid. }
		\label{Fig_Ret_Illus1_All_illust2} 
	\end{figure}

	In this case the best retrospective design do not have any dependency on the values of the covariance parameters as $|\SetUFour| = 1$ unlike the worst retrospective design. For three {sets} of values of {the} covariance parameter, efficiencies of various designs are given in Table~\ref{Efficence_Ills_2} . 
	\begin{table}[H]
	\begin{center}
		{
	\begin{tabular}{|l|cc| c c |}
	\hline
	$(\alpha, \beta)$
	& $eff(\xiboldMinus_{a5-worst}: \xibold )$  
	& $eff(\xiboldMinus_{a4}: \xibold )$
	& $eff(\xiboldMinus_{a5-worst}: \xiboldMinus_{{eq}_{ 7\times 3}} )$
		& $eff(\xiboldMinus_{a4} : \xiboldMinus_{{eq}_{ 7\times 3}} )$\\
		\hline 
			(.5, .7) & 0.4973 & 0.9415 & 0.5116 & 0.9687 \\
			(1, 5)   & 0.8133 & 0.9884 & 0.8168 & 0.9927  \\
			(10, 15) &  0.9335 & 0.9757 & 0.9559&   0.9991 \\
			\hline
		\end{tabular}	
		}
			\end{center}
			\caption{Efficiencies of best and worst designs obtained with respect to original design $\xibold$ and equispaced design $\xiboldMinus_{{eq}_{7 \times 3}}$}
			\label{Efficence_Ills_2}
			\end{table}
		
		From {the} values of $eff( \xiboldMinus_{a4}: \xibold)$ in Table \ref{Efficence_Ills_2}, we see that even after deleting points the efficiency of the new design is quite close to the initial design (which had {many} more points than the new design). From the values of $eff(\xiboldMinus_{a4}: \xiboldMinus_{{eq}_{7 \times 3}})$ we {see} that the best possible design is very close to the optimal equispaced design. Also, it is important to note that {a poor} choice {of retrospective design} could lead to a considerable loss in efficiency ({evidenced by the efficiencies of the} worst possible designs, $eff(\xiboldMinus_{a5-worst}: \xibold$).

	\begin{illustration} \label{Paper_II_Ilustration_3}
		Consider the example of monitoring methane flux ($Z(\cdot)$) in {the} troposphere from \cite{baran2015optimal}, where the covariance structure is considered to be separable exponential. The initial monitoring grid is taken to be an $8 \times 8 $ equispaced grid as in \cite{baran2015optimal}. If due to budget constraints the design needs to be reduced to a $6\times 5$ grid design, the proposed methodology to find the best possible design after deletion of points is used. 
	\end{illustration}
		{We} use an ordinary kriging model as in \cite{baran2015optimal}. Algorithms~\ref{PaperII_Algorithm3} and \ref{PaperII_Algorithm4} are used to find the best and worst possible retrospective grid designs, respectively. It is given that $n = 8, m = 8 $ and $n_{1}^\prime = 2 , m_{1}^\prime = 3$, as the final design is a $6\times 5$ grid. 
		
		In Step 1 {of} Algorithm~\ref{PaperII_Algorithm3}, see that the initial choice set for retrospective designs has 300 choices, that is $|\SetUFourPrime| = 300$. After running Step 2, $|\SetUFour| = 1$. So, there is only one choice for best possible retrospective design. Therefore, this best possible design does not depend on the covariance parameters. As above, denote the best possible design by $\xiboldMinus_{a4}$. For obtaining the worst possible design, Algorithm~\ref{PaperII_Algorithm4} is used. After Step 2, $|\SetUFiveWorst| = 1$. {We} denote the worst possible design by $\xiboldMinus_{a5-worst} $. Denote the $6\times 5$ equispaced design by $\xiboldMinus_{{eq}_{ 6\times 5}}$. The retrospective designs $\xiboldMinus_{a4}$ and $\xiboldMinus_{a5-worst}$ are obtained and compared with $\xibold$ and $\xiboldMinus_{{eq}_{ 6\times 5}}$ as shown in Figure~\ref{Fig_Ret_Illus1_All_illust3}.
		
		\begin{figure}[H]
			\centering
			\begin{subfigure}{.5\textwidth}
			\centering
			\includegraphics[width=.9\linewidth]{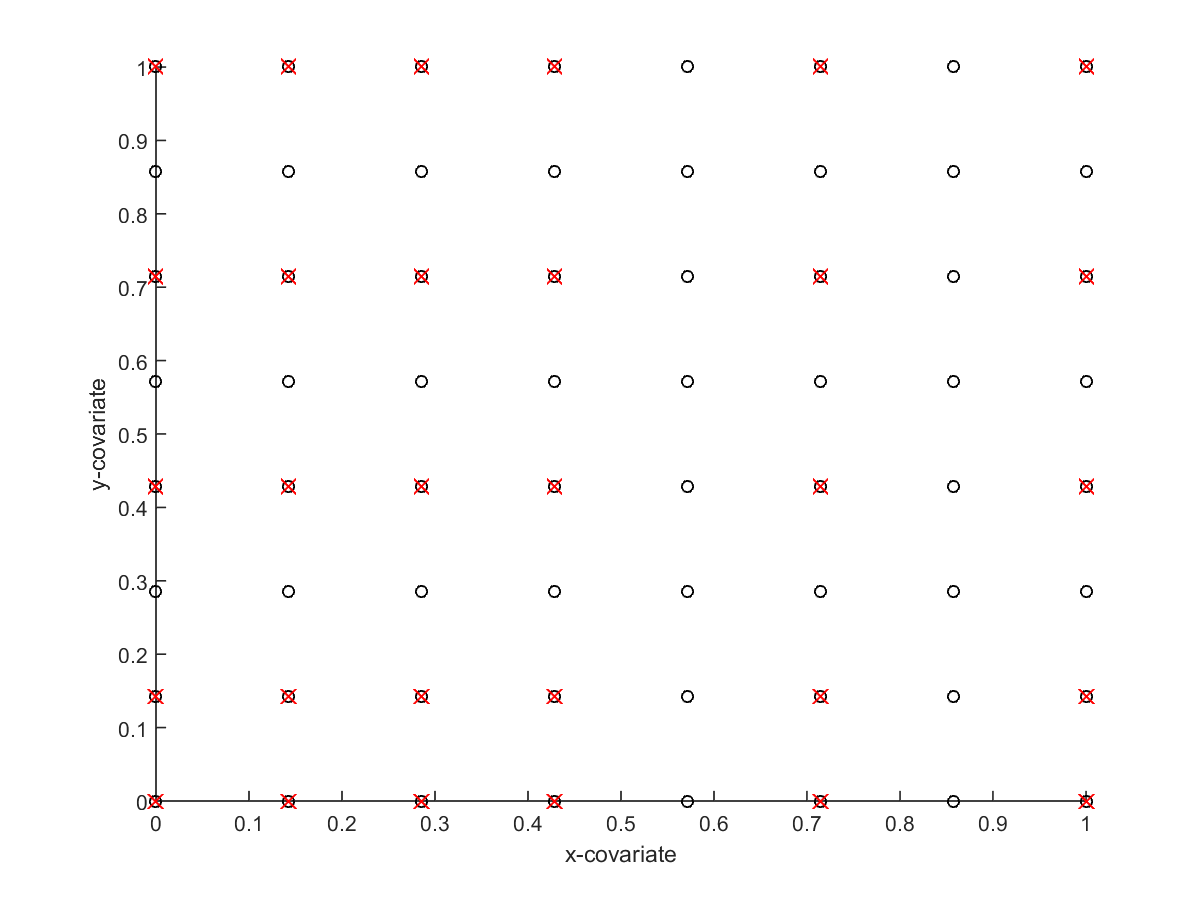}
		\caption{\scriptsize $\xiboldMinus_{a4}$ Vs $\xibold$    }
		\label{fig:ImgIllust3_1}
		\end{subfigure}%
		\begin{subfigure}{.5\textwidth}
		\centering
			\includegraphics[width=.9\linewidth]{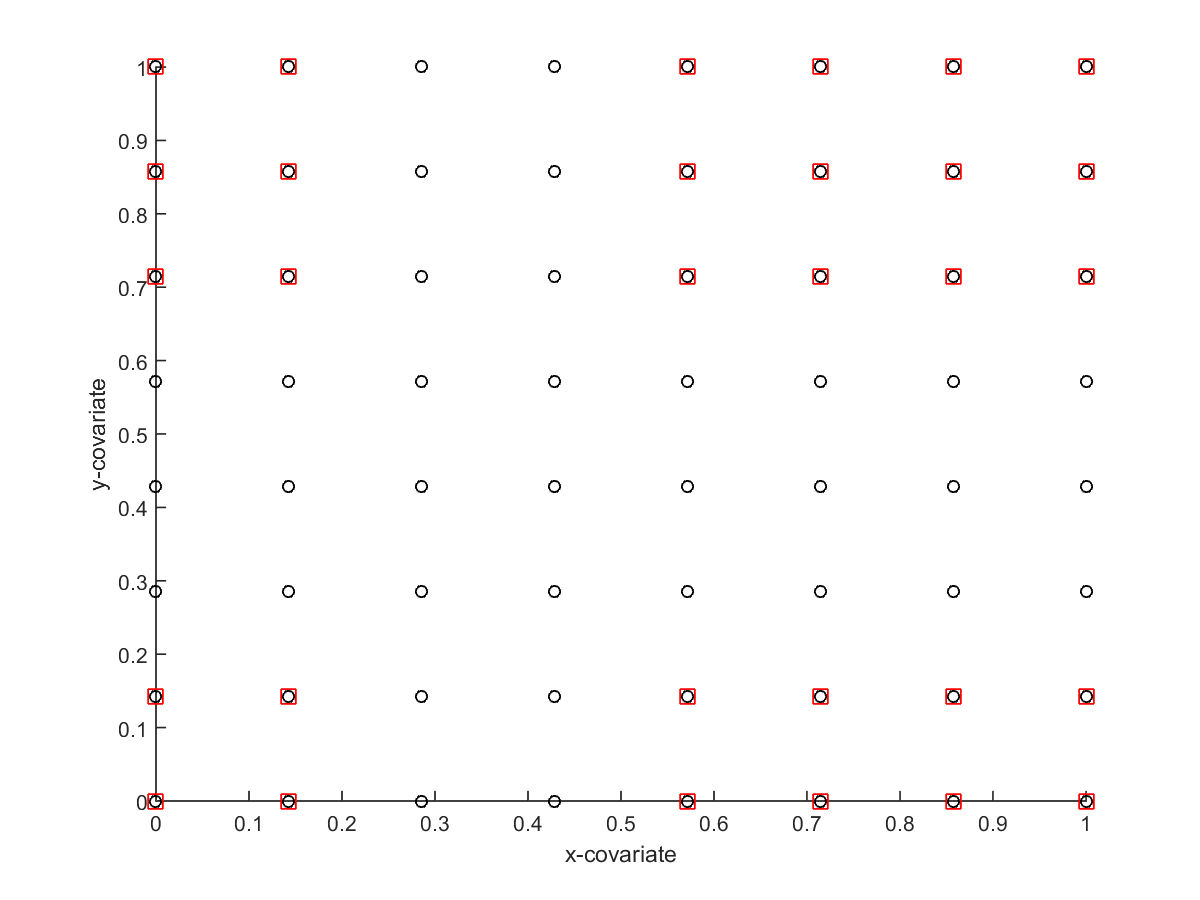}
			\caption{\scriptsize $\xiboldMinus_{a5-worst}$ Vs $\xibold$   }
			\label{fig:ImgIllust3_2}
		\end{subfigure}
	\vspace{1em}
		\centering
			\begin{subfigure}{.5\textwidth}
				\centering
					\includegraphics[width=.9\linewidth]{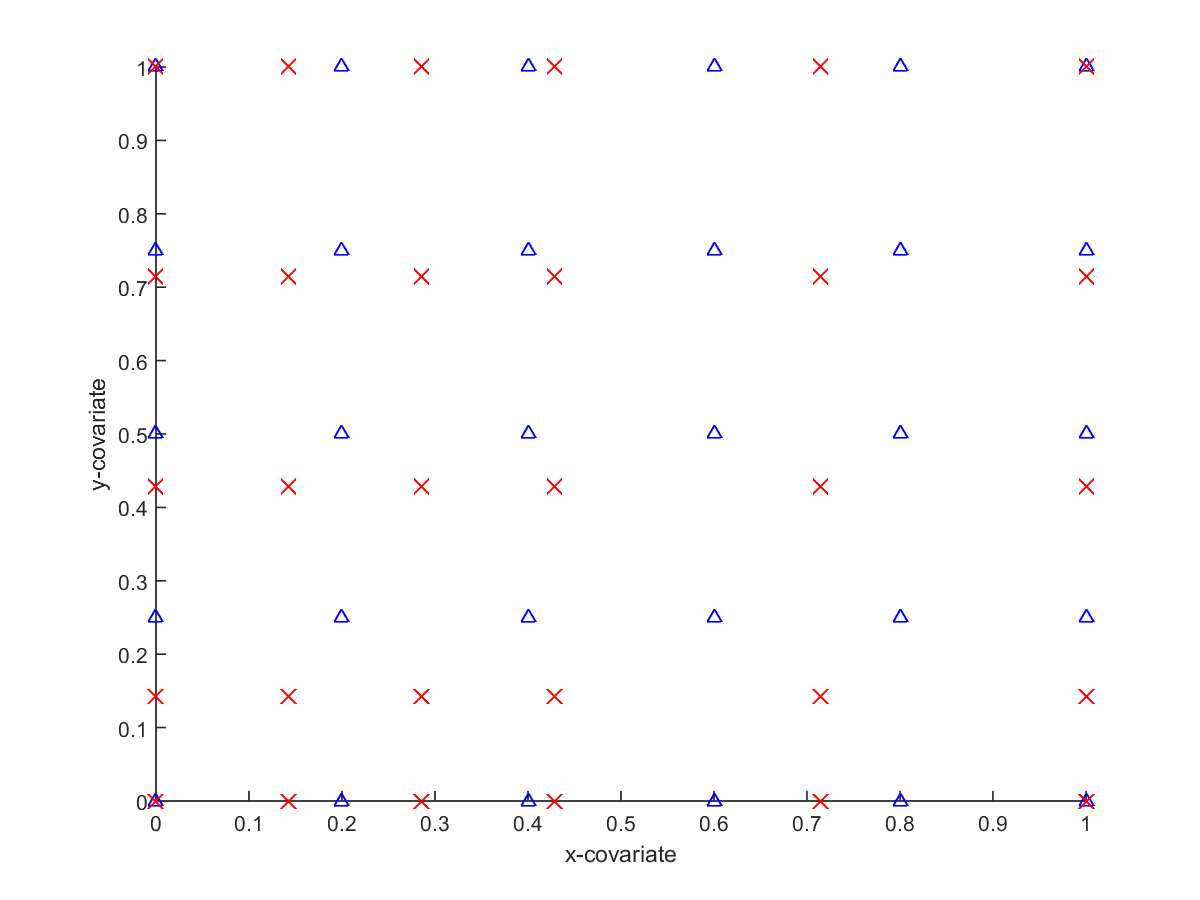}
					\caption{\scriptsize $\xiboldMinus_{a4}$  Vs $\xiboldMinus_{{eq}_{ 6\times 5}}$ }
					\label{Fig_Ret_Illus3_3}
					\end{subfigure}%
					\begin{subfigure}{.5\textwidth}
					\centering
					\includegraphics[width=.9\linewidth]{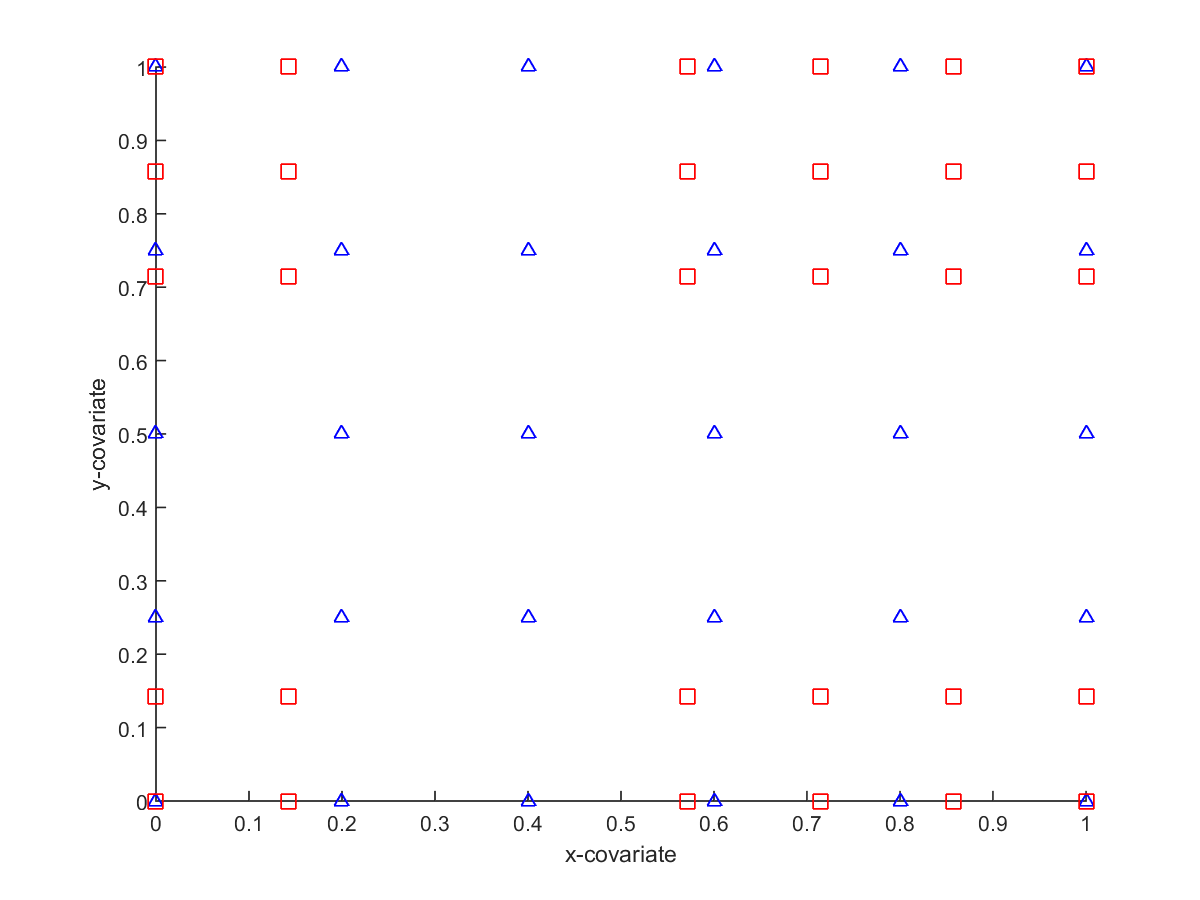}
					\caption{\scriptsize $\xiboldMinus_{a5-worst}$ Vs $\xiboldMinus_{{eq}_{ 6\times 5}}$    }
					\label{Fig_Ret_Illus3_4}
					\end{subfigure}
					\vspace{1em}
				\centering
		\caption{Comparison of design grids. 
		`$\color{red} \times$' - $\xiboldMinus_{a4}$: Best possible $6\times 5$ design. 	 		
		`$\color{red} \square$' - $\xiboldMinus_{a5-worst}$: Worst possible $6\times 5$ grid. 
	`$\circ$' - $\xibold$: Original $8 \times 8$ design grid.  `$\color{blue}\triangle$' - $\xiboldMinus_{eq_{ 6\times 5}}$: Equispaced $6\times 5$ grid. }	
	\label{Fig_Ret_Illus1_All_illust3} 
	\end{figure}
	In this case{, the best and worst possible designs do} not depend upon the value of {the} covariance parameters. However, for the purpose of illustration, in Table~\ref{Efficence_Ills_3} we take three {sets} of values of covariance parameters (as used above) to compare efficiencies of designs using $SMSPE$ values. 
	\begin{table}[H]
	\begin{center}
	{
	\begin{tabular}{|l|cc| c c |}
	\hline
	$(\alpha, \beta)$
	& $eff( \xiboldMinus_{a5-worst} : \xibold)$  
	& $eff(\xiboldMinus_{a4}: \xibold )$
	& $eff( \xiboldMinus_{a5-worst} : \xiboldMinus_{{eq}_{ 6\times 5}})$
	& $eff(\xiboldMinus_{a4} : \xiboldMinus_{{eq}_{ 6\times 5}} )$\\
	\hline 
		(.5, .7) & 0.2959 & 0.5117 & 0.4689 & 0.8108 \\
		(1, 5)   & 0.3962 & 0.5756 & 0.6142 & 0.8922 \\
		(10, 15) &  0.8815 & 0.8962 &  0.9559 &   0.9838 \\
		\hline
	\end{tabular}	
	}
	\end{center}
	\caption{Efficiencies of designs obtained with respect to original design and equispaced design.}
	\label{Efficence_Ills_3}
		\end{table}

		From values of $eff( \xiboldMinus_{a4} : \xibold)$ in Table~\ref{Efficence_Ills_3}, we can say that the new design $\xiboldMinus_{a4}$ has {considerably} reduced accuracy (in terms of \textit{SMSPE}). This could be attributable to the fact that in this case, the starting design $\xibold$ was the optimal for size $8\times 8$ (as it is equispaced). However, if we look at the values of $eff(\cdot : \xibold  )$ and $eff(\cdot: \xiboldMinus_{{eq}_{ 6\times 5}})$, it {can} be seen that the best possible design has much better efficiency than the worst possible retrospective design. So, identifying the points to remove is important in order to ensure the efficiency of the newly obtained design. 
		
	\section*{{Software} used for computation}
	We have implemented Lemmas~\ref{Paper_II_Minorised_Designs} and \ref{Paper_II_Majorised_Designs},   Algorithms~\ref{PaperII_Algorithm1},  \ref{PaperII_Algorithm2},  \ref{PaperII_Algorithm3} and  \ref{PaperII_Algorithm4} in $MATLAB-R2015a$ software. The program codes are available on request from the first author.
	
	\section{Concluding remarks}\label{Paper_II_Conclusion}
	Two-dimensional equispaced grid {designs} are proved to be locally and Bayesian (prospective) G-optimal for simple and ordinary kriging, when the covariance structure is taken to be separable and exponential.  
	
	A more realistic scenario of finding {the} best possible retrospective designs is discussed. The criterion for evenness of designs is proposed for the purpose of comparing two-dimensional grid designs. Then, the mathematical relationship between the design and $SMSPE$ is studied under the proposed criterion and it is shown that indeed, a more evenly spaced grid will lead to lower values of $SMSPE$. Deterministic algorithms for finding the best possible retrospective designs by adding and also deleting points {are} given. {As} constructing a new design by deleting points from {an} existing design is much easier (there are {only finitely many} choices for the new design) compared to {the} problem of addition of new points to an existing (there are {infinitely many} choices), the {latter} algorithms require more mathematical {effort}. The convergence of the algorithms is ensured by reducing the cardinality of the initial choice set to a finite number. The algorithms are designed in such a way that for {the purpose of} determining the best possible retrospective design, in many cases the covariance parameter values are not required (as shown in the illustrations). \\
	
	\noindent The contributions of this paper are:
	\begin{itemize}
		\item theoretically finding  prospective G-optimal grid designs {under} both frequentist and Bayesian paradigms,  
		\item providing a criterion for comparing evenness of two-dimensional grid designs and then proving that under this criterion a more evenly spread design leads to a lower value of $SMSPE$,
		\item providing two deterministic algorithms for finding {the} best possible retrospective designs (with respect to \textit{SMSPE} criteria) by adding points sequentially or simultaneously to an existing design, 
		\item providing a deterministic algorithm for finding {the} best possible retrospective design (with respect to \textit{SMSPE} criteria) by simultaneously deleting points from an existing design. 
	\end{itemize}
	
	In this article, we found the retrospective designs by adding or deleting a pre-specified number of points. We want to extend this work further by finding the optimal number of design points for a simple or ordinary kriging model. In a retrospective design framework, that would mean investigating the number of points that need to be added or deleted from an existing design. We further want to extend this work by finding prospective optimal design and retrospective best possible designs for universal kriging models. 
	\begin{appendix} 
			\section{Appendix} \label{Appendix} \label{Appendix_IMSPE_SCKandOCK} \label{AppendixA} \label{AppendixC} \label{PaperIIAppendixA}

			\noindent	\textbf{A.1}
			The covariance matrices $\MatrixP$ and $\MatrixQ$ are exponential covariance matrices, so from \cite{Kriging_Antognini_Zagoraiou_2010} we have: 
			\begin{align}
				\VectOneN^{T} \MatrixP^{-1} \VectOneN &= 1 +  \sum_{i=1}^{n-1} \dfrac{e^{\alpha d_{i}} - 1}{e^{\alpha d_{i}} + 1 } \label{PaperII_Loc_Eq1} \text{ and }\\
				\VectOneM^{T} \MatrixQ^{-1} \VectOneM &= 1 +  \sum_{j=1}^{m-1} \dfrac{e^{\beta \delta_{j}} - 1}{e^{\beta \delta_{j}} + 1 } \label{PaperII_Loc_Eq2}. 
			\end{align}
			{Whenever} $\displaystyle{\sum_{i=1}^{n-1} d_{i} = 1}$ and $\displaystyle{\sum_{j=1}^{m-1} \delta_{i} = 1}$,
			\begin{align*}
			\VectOneN^{T} \MatrixP^{-1} \VectOneN &=  \sum_{i=1}^{n-1} d_{i} + \dfrac{e^{\alpha d_{i}} - 1}{e^{\alpha d_{i}} + 1 } \text{ and } \\
			\VectOneM^{T} \MatrixQ^{-1} \VectOneM &= \sum_{j=1}^{m-1} \delta_{i} + \dfrac{e^{\beta \delta_{j}} - 1}{e^{\beta \delta_{j}} + 1 }. 
			\end{align*}
				From \cite{Paper_One_2021_arxiv} we can say $\dfrac{1}{\VectOneN^{T} \MatrixP^{-1} \VectOneN}$ and $\dfrac{1}{\VectOneM^{T} \MatrixQ^{-1} \VectOneM}$  are Schur-convex functions in $d_{i}'s$ and $\delta_{j}'s$, respectively and {are} minimized for equispaced  $d_{i}$ and $\delta_{j}$, respectively for $i = 1,\ldots, n-1$ and $j = 1,\ldots,m-1$. \\
				
			\noindent	\textbf{A.2} 
			{Consider} $x_{0} \in [x_{i}, x_{i+1}]$ and $y_{0} \in [y_{j}, y_{j+1}]$  for some $ i = 1,\ldots,n-1$, $ j = 1,\ldots,m-1$. Define $ a = x_{0} - x_{i}$ and $ b = y_{0} - y_{j}${. Then} 
			\begin{align}
			\sigmaPNot^{T} \MatrixP^{-1} \sigmaPNot &= \dfrac{ e^{-2 \alpha a } -2  e^{-2\alpha d_{i}   } + e^{-2 \alpha (d_{i} - a ) }   }{ 1- e^{-2 \alpha d_{i} }   },    \label{PaperII_decomposition1} \\ 
			\VectOneN^{T} \MatrixP^{-1} \sigmaPNot  &= \dfrac{e^{- \alpha a } + e^{- \alpha(d_{i}  - a )  } }{ 1 + e^{- \alpha d_{i}  }}. \label{PaperII_decomposition2}\\
					\sigmaQNot^{T} \MatrixQ^{-1} \sigmaQNot &= \dfrac{ e^{-2\beta b } -2  e^{-2\beta d_{i}   } + e^{-2\beta (d_{i} - b ) }   }{ 1- e^{-2 \beta d_{i} }   },  \text{ and }  \label{PaperII_decomposition3} \\ 
			\VectOneN^{T} \MatrixQ^{-1} \sigmaQNot  &= \dfrac{e^{- \beta b } + e^{- \beta (d_{i}  - b )  } }{ 1 + e^{- \beta d_{i}  }}. \label{PaperII_decomposition4}		
			\end{align}
					{Detailed calculations are provided in \cite{Paper_One_2021_arxiv}.}
	\section{Appendix} 		\label{AppendixG}
The properties of components of $SMSPE_{sk}$ and $SMSPE_{ok}$ are studied here. These properties would be used for determining properties of the optimal design.
	\begin{result} \label{Result1}
		If $(x_{0},y_{0}) \in [x_{i}, x_{i+1}]\times[y_{j},y_{j+1}]  $ for some $i = 1,\ldots,n$ and $j= 1,\ldots,m${, then}  the function $	\Big(1- \; \sigmaPNot^{T} \MatrixP^{-1} \sigmaPNot \;\; \sigmaQNot^{T}  \MatrixQ^{-1} \sigmaQNot \Big)$ in variables $(x_{0},y_{0})$ attains its global {maximum} at $(x_{0},y_{0}) = (x_{i}+d_{i}/2, y_{j}+ \delta_{j}/2)$. Also, $  \sup_{(x_{0}, y_{0}) \in [x_{i},x_{i+1}]\times[y_{j}, y_{j+1}]}  \sigma^{2}_{sk}(x_{0}, y_{0}) = \sigma^{2}  \Big( 1-  \dfrac{2 }{e^{\alpha d_{i}} +1 } \;  \dfrac{2 }{e^{\beta \delta_{j}} +1 } \Big)$.
	\end{result}
	
		\begin{proof}
		Take the functions $f(\cdot)$ and $g(\cdot)$ defined over $[x_{i},x_{i+1}]$ and $[y_{j},y_{j+1}]$ respectively, such {that} $f(x_{0}) =  \sigmaPNot^{T}  \MatrixP^{-1} \sigmaPNot$ and $g(y_{0}) =  \sigmaQNot^{T}  \MatrixQ^{-1} \sigmaQNot$. Define the function $h \colon [x_{i},x_{i+1}]\times[y_{j},y_{j+1}] \to \mathbb{R} $ such that $h(x,y) = 1- f(x) g(y) $. 
	{Then} $h(x_{0},y_{0}) = 	\Big(1- \; \sigmaPNot^{T} \MatrixP^{-1} \sigmaPNot \;\; \sigmaQNot^{T}  \MatrixQ^{-1} \sigmaQNot \Big)$. Consider \eqref{PaperII_decomposition1} and \eqref{PaperII_decomposition3} for further calculation. From Proposition 4.4 and Lemma 4.2 in \cite{ghorpade2010course}, for a continuous function $h(x,y)$, the {global} maximum is attained either {at} a critical point or at a boundary point. We see that for $(x^{\prime},y^{\prime}) \in (x_{i},x_{i+1})\times(y_{j},y_{j+1})$, if $\nabla h(x^{\prime},y^{\prime})\Big| = (0,0)$, then $(x^{\prime},y^{\prime}) = (x_{i}+d_{i}/2, y_{j}+ \delta_{j}/2)$. So, $ (x_{i}+d_{i}/2, y_{j}+ \delta_{j}/2)$ is a critical point for $h(x,y)$. Now we have,
		\begin{align*}
		\pmb \Delta(x,y) &= 
		\begin{vmatrix}
		h_{xx}(x,y) & h_{xy}(x,y)\\
		h_{yx}(x,y) & h_{yy}(x,y)
		\end{vmatrix}\\
		&=
		\begin{vmatrix}
		-f_{xx}(x)g(y) & -f_{x}(x)g_{y}(y)\\
		-f_{x}(x)g_{y}(y) & -f(x)g_{yy}(y)
		\end{vmatrix}			
		\end{align*}
		We see that $h_{xx}(x,y)\Big|_{(x_{i}+d_{i}/2, y_{j}+ \delta_{j}/2)} < 0 $ and $\pmb \Delta(x,y) |_{(x_{i}+d_{i}/2, y_{j}+ \delta_{j}/2)} = 	f(x_{i}+d_{i}/2)g(y_{j}+ \delta_{j}/2)f_{xx}(x_{i}+d_{i}/2)g_{yy}(y_{j}+ \delta_{j}/2) > 0 .$ Hence, using the determinant rule, $(x_{i}+d_{i}/2, y_{j}+ \delta_{j}/2)$ is a point of {global maximum} for $h(x,y)$ over $(x_{i},x_{i+1})\times(y_{j},y_{j+1})$. When we check the values of $h(x,y)$ {at} the boundary, we see that the functions $h(x_{i}, y)$ and {$h(x_{i+1}, y)$} defined over $[y_{j}, y_{j+1}]$ {attain} maxima at $y_{j}+ \delta_{j}/2${, and} $h(x, y_{j})$ and $h(x, y_{j+1})$ defined over $[x_{i}, x_{i+1}]$ {attain} maximum at $x_{i}+d_{i}/2$. When we compare the function values {at critical points and at the} boundary, we find that $(x_{i}+d_{i}/2, y_{j}+ \delta_{j}/2)$ is {the maximum} for $h(x,y)$ over $[x_{i},x_{i+1}]\times[y_{j},y_{j+1}]$. So,
			\begin{align*}
		\sigma^{2} \sup_{(x_{0}, y_{0}) \in [x_{i},x_{i+1}]\times[y_{j},y_{j+1}]} \Big( 1- \sigmaQNot^{T}  \MatrixQ^{-1} \sigmaQNot \;\; \sigmaPNot^{T}  \MatrixP^{-1} \sigmaPNot \Big) 
		&= \sigma^{2} \sup_{(x_{0}, y_{0}) \in [x_{i},x_{i+1}]\times[y_{j},y_{j+1}]} h(x_{0},y_{0})\\
	    &= \sigma^{2} \Big( 1- f(x_{i}+d_{i}/2) g(y_{j}+ \delta_{j}/2)\Big) \\
		&= \sigma^{2} \Big( 1- 2 \dfrac{e^{\alpha d_{i}} -1 }{e^{2\alpha d_{i}} -1 } \; 2 \dfrac{e^{\beta \delta_{j}} -1 }{e^{2\beta \delta_{j}} -1 } \Big) \\
		&= \sigma^{2} \Big( 1-  \dfrac{2 }{e^{\alpha d_{i}} +1 } \;  \dfrac{2 }{e^{\beta \delta_{j}} +1 } \Big).
		\end{align*}	
	\end{proof}

	\begin{result} \label{Result3}
				If $(x_{0},y_{0}) \in [x_{i}, x_{i+1}]\times[y_{j},y_{j+1}]  $ for some $i = 1,\ldots,n$ and $j= 1,\ldots,m${, then} the function 	$\Big(1-\sigmaQNot^{T} \MatrixQ^{-1} \VectOneM \;\; \sigmaPNot^{T} \MatrixP^{-1} \VectOneN \Big)^2$ in variables $(x_{0},y_{0})$, attains its {global maximum} at $(x_{0},y_{0}) = (x_{i}+d_{i}/2, y_{j}+ \delta_{j}/2)$. Also, $	\sup_{(x_{0}, y_{0}) \in [x_{i},x_{i+1}]\times[y_{j},y_{j+1}]}  \Big(1-\sigmaQNot^{T} \MatrixQ^{-1} \VectOneM \;\; \sigmaPNot^{T} \MatrixP^{-1} \VectOneN \Big)^2 =
	\Big( 1- \dfrac{2 e^{-\alpha d_{i}/2}}{e^{- \alpha d_{i}} +1 } \dfrac{2 e^{-\beta\delta_{i}/2}}{e^{- \beta \delta_{i}} +1 } \Big)^2. $
			\end{result}
	\begin{proof}  
	Take the functions $f_{1}(\cdot)$ and $g_{1}(\cdot)$ defined over $[x_{i},x_{i+1}]$ and $[y_{i},y_{i+1}]$ respectively, such that, $f_{1}(x_{0}) = \sigmaPNot^{T} \MatrixP^{-1} \VectOneN $ and $g_{1}(y_{0}) = \sigmaQNot^{T} \MatrixQ^{-1} \VectOneM$. Define the function $	h_{1} \colon [x_{i},x_{i+1}]\times[y_{j},y_{j+1}] \to \mathbb{R} $ such that, $h_{1}(x,y) = (1- f_{1}(x) g_{1}(y))^{2}$. {Therefore,} $h_{1}(x,y) = \Big(1-\sigmaQNot^{T} \MatrixQ^{-1} \VectOneM \;\; \sigmaPNot^{T} \MatrixP^{-1} \VectOneN \Big)^2$. We use \eqref{PaperII_decomposition2} and \eqref{PaperII_decomposition4} for further calculations. Again we use Proposition 4.4 and Lemma 4.2 in \cite{ghorpade2010course}, that a continuous function attains {its global} maximum either {at} a critical point or a boundary point. We check that for $(x^{\prime},y^{\prime}) \in (x_{i},x_{i+1})\times(y_{j},y_{j+1})$, if $\nabla h_{1}(x^{\prime},y^{\prime})\Big| = (0,0)$, then $(x^{\prime},y^{\prime}) = (x_{i}+d_{i}/2, y_{j}+ \delta_{j}/2)$. So, $(x^{\prime},y^{\prime}) = (x_{i}+d_{i}/2, y_{j}+ \delta_{j}/2)$ is a critical point for $h_{1}(x,y)$.  We have,
		\begin{align*}
		\pmb \Delta_{1}(x,y) &= 
		\begin{vmatrix}
		{h_{1}}_{xx}(x,y) & {h_{1}}_{xy}(x,y)\\
		{h_{1}}_{yx}(x,y) & {h_{1}}_{yy}(x,y)
		\end{vmatrix}\\
		&=
		\begin{vmatrix}
		-{f_{1}}_{xx}(x){g_{1}}(y) & -{f_{1}}_{x}(x){g_{1}}_{y}(y)\\
		-{f_{1}}_{x}(x){g_{1}}_{y}(y) & -{f_{1}}(x){g_{1}}_{yy}(y)
		\end{vmatrix}			
		\end{align*}
		If we use the determinant rule, then ${h_{1}}_{xx}(x,y)|_{(x_{i}+d_{i}/2, y_{j}+ \delta_{j}/2)} < 0$ and $\pmb \Delta_{1}(x,y) |_{(x_{i}+d_{i}/2, y_{j}+ \delta_{j}/2)} = 	f(x_{i}+d_{i}/2)g(y_{j}+ \delta_{j}/2)f_{xx}(x_{i}+d_{i}/2)g_{yy}(y_{j}+ \delta_{j}/2) > 0 .$ Hence, $(x_{i}+d_{i}/2, y_{j}+ \delta_{j}/2)$ is a point of {global maximum} for $h_{1}(x,y)$ over $(x_{i},x_{i+1})\times(y_{j},y_{j+1})$. Checking the values of $h_{1}(x,y)$ {at the} boundary as in Result \ref{Result1} and comparing the maxima {at the} boundary with the critical point we find that the {global maximum} of $h_{1}(x,y)$ over $[x_{i},x_{i+1}]\times[y_{j},y_{j+1}]$ is attained at $(x_{i}+d_{i}/2, y_{j}+ \delta_{j}/2)$. So,
				\begin{align*}
			\sup_{(x_{0}, y_{0}) \in [x_{i},x_{i+1}]\times[y_{j},y_{j+1}]}  \Big(1-\sigmaQNot^{T} \MatrixQ^{-1} \VectOneM \;\; \sigmaPNot^{T} \MatrixP^{-1} \VectOneN \Big)^2 
			& = \sup_{(x_{0}, y_{0}) \in [x_{i},x_{i+1}]\times[y_{j},y_{j+1}]}  h_{1}(x,y) \\
			&= \sup_{(x_{0}, t_{0}) \in [x_{i},x_{i+1}]\times[y_{j},y_{j+1}]} \Big(1-f_{1}(x_{0}) g_{1}(y_{0})\Big)^{2} \\
			&= \Big(1-f_{1}(x_{i}+d_{i}/2) \; g_{1}(y_{j}+ \delta_{j}/2)\Big)^{2} \\
			&=    \Big( 1- \dfrac{2 e^{-\alpha d_{i}/2}}{e^{- \alpha d_{i}} +1 } \dfrac{2 e^{-\beta\delta_{i}/2}}{e^{- \beta \delta_{i}} +1 } \Big)^2 
			\end{align*} 
	\end{proof}
	
	\section{Appendix} 	 \label{AppendixGMinus1}
	Proof of Theorems~\ref{Paper_II_Theorem_ST_SK_SMSPE}  and \ref{Paper_II_Theorem_ST_SK_OK_SMSPE_PseudoBAyesian} in Section~\ref{Paper_II_Prospective_Design} are provided here. \\

	\noindent Following is the proof for simple kriging case in Theorem~\ref{Paper_II_Theorem_ST_SK_SMSPE}. 
	\begin{proof} \label{Paper_II_Theorem_ST_SK_SMSPE_Proof}
		From Result~\ref{Result1} in Appendix \ref{AppendixG} we have, 
		\begin{align*}
		SMSPE_{sk}(\xibold) 
		&= \sigma^{2} \max_{\substack{i=1,\ldots,n-1\\j=1,\ldots,m-1}}\sup_{(x_{0}, y_{0}) \in [x_{i},x_{i+1}]\times[y_{j},y_{j+1}]} \sigma^{2}_{sk}(x_{0}, y_{0})\\
		&= \sigma^{2} \max_{\substack{i=1,\ldots,n-1\\j=1,\ldots,m-1}} 1-  \dfrac{2 }{e^{\alpha d_{i}} +1 } \;  \dfrac{2 }{e^{\beta \delta_{j}} +1 } \\
		&= \sigma^{2} \max_{\substack{i=1,\ldots,n-1\\j=1,\ldots,m-1}} \Big( 1-  s_{\alpha}(d_{i}) \; s_{\beta}(\delta_{j}) \Big),
		\end{align*}	
		where function $ s_{\theta}(\zeta) =  \dfrac{2}{e^{\theta \zeta} +1 }$ over 
		$\mathbb{R}$ is decreasing in $\zeta$ as $s_{\theta}^{\prime}(\zeta) < 0 $ for $\theta = \alpha, \beta$. So,   
		\begin{align}
		SMSPE_{sk}(\xibold) 
		&= \sigma^{2} \Big( 1-  s_{\alpha}( \norm{\VectD}_{\infty} ) \; s_{\beta}(\norm{\VectDelta}_{\infty}) \Big). \label{Paper_II_SMSPE_SK}
		\end{align}	
		We next minimize the \textit{SMSPE} {to obtain} the optimal design,
		\begin{align}
		\min_{ \substack{  \{d_{1},\ldots,d_{n} \\ \delta_{1},\ldots,\delta_{n}\}  } } SMSPE_{sk}(\xibold)  
		& = \min_{ \substack{  \{d_{1},\ldots,d_{n} \\ \delta_{1},\ldots,\delta_{n}\}  } }
		\sigma^{2} \Big( 1-  s_{\alpha}( \norm{\VectD}_{\infty} ) \; s_{\beta}( \norm{\VectDelta}_{\infty} ) \Big) \label{Loc_expression_1}.
		\end{align}
		The expression in \eqref{Loc_expression_1} is minimized when $\norm{\VectD}_{\infty}$ and $ 
		\norm{\VectDelta}_{\infty}$ 
		are minimized, as $s_{\theta}(\cdot)$ is a decreasing function in $\zeta$. 
		Clearly, the minimum value is attained when both $\norm{\VectD}_{\infty} $ and $
		\norm{\VectDelta}_{\infty}$ {arise} from an equispaced design, which proves that an equispaced grid design is G-optimal.
	\end{proof}
	
	\vspace{2em}
	\noindent Following is the proof for ordinary kriging case in Theorem~\ref{Paper_II_Theorem_ST_SK_SMSPE}
	\begin{proof}
		If $(x_{0}, y_{0}) \in [x_{i},x_{i+1}]\times[y_{j},y_{j+1}]$ for some $i = 1,\ldots,n$ and $j = 1,\ldots,m$ then from Results~\ref{Result1} and \ref{Result3} in Appendix \ref{AppendixG} we know that $
		\Big(  1 - \sigmaQNot^{T}  \MatrixQ^{-1} \sigmaQNot \;\; \sigmaPNot^{T}  \MatrixP^{-1} \sigmaPNot \Big) $ and  $ \Big(1-\sigmaQNot^{T} \MatrixQ^{-1} \VectOneM \;\; \sigmaPNot^{T} \MatrixP^{-1} \VectOneN \Big)^2$ both attain a supremum at $(x_{i}+d_{i}/2, y_{j}+ \delta_{j}/2)$. Recall, $\Omega_{x}(\xibold) = \VectOneN^{T} \MatrixP^{-1} \VectOneN$ and $ \Omega_{y}(\xibold) = \VectOneM^{T} \MatrixQ^{-1} \VectOneM$. Then, using Results~\ref{Result1} and ~\ref{Result3} from Appendix \ref{AppendixG} and equation~\eqref{OrdinarykrigingMSPEII}, we have 
		\begin{align*}
		SMSPE_{ok}(\xibold)
		&= \max_{\substack{i=1,\ldots,n-1\\j=1,\ldots,m-1}}\sup_{(x_{0}, y_{0}) \in [x_{i},x_{i+1}]\times[y_{j},y_{j+1}]} \sigma^{2}_{ok}(x_{0}, y_{0})\\
		&= \sigma^{2}    \max_{\substack{i=1,\ldots,n-1\\j=1,\ldots,m-1}} \Bigg[ 1-  s_{\alpha}(d_{i}) \; s_{\beta}(\delta_{j})  + \dfrac{1}{\Omega_{x}(\xibold) } \; \dfrac{1}{\Omega_{y}(\xibold) } \Big( 1- u_{\alpha}(d_{i}) u_{\beta}(\delta_{i})\Big)^2 \Bigg],
		\end{align*} 
		where the function $ u_{\theta }(\zeta) =  \dfrac{2 e^{-\theta \zeta/2}}{e^{- \theta \zeta} +1 }$  over $\mathbb{R}$ is decreasing in $\zeta$ as $u_{\theta}^{\prime}(\zeta) < 0 $, for $\theta = \alpha, \beta$. Since $s_{\theta}(\cdot)$ and $u_{\theta}(\cdot)$ are decreasing in $\zeta$, {we} can say
		\begin{align}
		SMSPE_{ok}(\xibold)
		&= \sigma^{2}   \Bigg[ 1-  s_{\alpha}( \norm{\VectD}_{\infty} ) \; s_{\beta}(\norm{\VectDelta}_{\infty})  \nonumber \\ & 
		\;\;\;\;\;\;\;\;\;\;\;\;\;\;\;\;\;\;\;\;\;\;\;\;\;\;\;\;\;\; + \dfrac{1}{\Omega_{x}(\xibold) } \; \dfrac{1}{\Omega_{y}(\xibold) } \Big( 1- u_{\alpha}(\norm{\VectD}_{\infty} ) u_{\beta}(\norm{\VectDelta}_{\infty})\Big)^2 \Bigg] \label{Paper_II_SMSPE_OK}.
		\end{align}		
		Minimizing $SMSPE_{ok}(\xibold)$ we get,   
		\begin{align*}
		\min_{ \substack{  \{d_{1},\ldots,d_{n-1} \\ \delta_{1},\ldots,\delta_{m-1}\}  } } SMSPE_{ok}(\xibold)  & = 
		\min_{ \substack{  \{d_{1},\ldots,d_{n-1} \\ \delta_{1},\ldots,\delta_{m-1}\}  } }
		\Bigg[ 1-  s_{\alpha}(\norm{\VectD}_{\infty} ) \; s_{\beta}(\norm{\VectDelta}_{\infty}) \nonumber \\
		& \;\;\;\;\;\;\;\;\;\;\;\;\;\;\;\;\;\;\;\;\;\;\;\;\;\;\;\;\;\; + \dfrac{1}{\Omega_{x}(\xibold) } \; \dfrac{1}{\Omega_{y}(\xibold) } \Big( 1- u_{\alpha}(\norm{\VectD}_{\infty} ) u_{\beta}(\norm{\VectDelta}_{\infty})\Big)^2 \Bigg].
		\end{align*}

		Since, $s_{\theta}(\cdot)$ and $u_{\theta}(\cdot)$ are decreasing functions in $\zeta$, therefore $\Big( 1-  s_{\alpha}(\norm{\VectD}_{\infty} ) \; s_{\beta}(\norm{\VectDelta}_{\infty}) \Big)$ and $\Big( 1- u_{\alpha}(\norm{\VectD}_{\infty} ) u_{\beta}(\norm{\VectDelta}_{\infty})\Big)^2 $ are minimized when $\norm{\VectD}_{\infty} $ and $\norm{\VectDelta}_{\infty}$ are minimized. Note, $\norm{\VectD}_{\infty} $ and $\norm{\VectDelta}_{\infty}$ both attain their respective minima for an equispaced design. Also, \cite{Paper_One_2021_arxiv} showns that $\displaystyle{\dfrac{1}{\VectOneN^{T} \MatrixP^{-1} \VectOneN} }$ and  $\displaystyle{ \dfrac{1}{\VectOneM^{T} \MatrixQ^{-1} \VectOneM } }$ are Schur-convex functions and hence minimized for an equispaced partition. Therefore, $SMSPE_{ok}$ is minimized for an equispaced partition, which shows that an equispaced design is G-optimal.
	\end{proof}
	
	\vspace{2em}
	\noindent	Following is the proof of Theorem~\ref{Paper_II_Theorem_ST_SK_OK_SMSPE_PseudoBAyesian} 
	\begin{proof}
		Consider that the distributions of $\sigma^{2}$, $\alpha$ and $\beta$ are given by the distribution functions $r_{1}(.)$, $r_{2}(.)$, and $r_{3}(.)$, respectively. Then,  
		\begin{align*}
		\mathcal{R}_{sk}(\xibold) &=   E_{r_{1}}[\sigma^{2}]  \int \int  \Big( 1-  s_{\alpha}(\norm{\VectD}_{\infty} ) \; s_{\beta}(\norm{\VectDelta}_{\infty}) \Big) \;\; r_{2}(\alpha) r_{3}(\beta) \;\; dr_{2} \; dr_{3}\\
		\mathcal{R}_{ok}(\xibold) &= E_{r_{1}}[\sigma^{2}]  \int \int   \Bigg( 1-  s_{\alpha}(\norm{\VectD}_{\infty} 
		) \; s_{\beta}(\norm{\VectDelta}_{\infty})   + \dfrac{( 1- u_{\alpha}(\norm{\VectD}_{\infty} 
			) u_{\beta}(\norm{\VectDelta}_{\infty}) )^2}{\Omega_{x}(\xibold) \; \Omega_{y}(\xibold)  }  \Bigg) \;\; r_{2}(\alpha) r_{3}(\beta) \;\; dr_{2} \; dr_{3}
		\end{align*}
		Hence, 
		\begin{align}
		\min_{\xibold} \mathcal{R}_{sk}(\xibold) 
		&=   E_{r_{1}}[\sigma^{2}] \;\;   \int \int \min_{\xibold} \Big( 1-  s_{\alpha}(\norm{\VectD}_{\infty} 
		) \; s_{\beta}(\norm{\VectDelta}_{\infty}) \Big)  \;\; r_{2}(\alpha) r_{3}(\beta) \;\; dr_{2} \; dr_{3} \label{key11}\\
		\min_{\xibold} 	\mathcal{R}_{ok}(\xibold) 
		&= E_{r_{1}}[\sigma^{2}] \;\; \int \int   \min_{\xibold} \Bigg[ 1-  s_{\alpha}(\norm{\VectD}_{\infty} 
		) \; s_{\beta}(\norm{\VectDelta}_{\infty})   \nonumber \\ 
		          & \;\;\;\;\;\;\;\;\;\;\;\;\;\;\;\;\;\;\;\;\;\;\;\;\;\;\;\;\;\;\;\;\;\; \;\;\;+\dfrac{( 1- u_{\alpha}(\norm{\VectD}_{\infty} 
			) u_{\beta}(\norm{\VectDelta}_{\infty}) )^2}{\Omega_{x}(\xibold) \; \Omega_{y}(\xibold)  }  \Bigg] \;\; r_{2}(\alpha) r_{3}(\beta) \;\; dr_{2} \; dr_{3} \label{key12}
		\end{align}
		Using Theorem~\ref{Paper_II_Theorem_ST_SK_SMSPE} 
		and equations \eqref{key11} and \eqref{key12} we can say that $\mathcal{R}_{sk}(\xibold)$ and $\mathcal{R}_{ok}(\xibold)$ are minimized for an equispaced grid design. 
	\end{proof}

	\section{Appendix} 	 \label{AppendixG0}
	Proof of Theorem~\ref{Paper_II_Theorem_Any_Design}
	\begin{proof} 
		{As} $ \VectD \prec \VectD^{\prime}$ and $ \VectDelta \prec \VectDelta^{\prime}$, therefore from the definition of majorization in \eqref{Paper_II_Majorization_Definition} we have $\norm{\VectD}_{\infty} \leq 	\norm{\VectD^{\prime}}_{\infty}  $ and $\norm{\VectDelta}_{\infty} \leq \norm{\VectDelta^{\prime}}_{\infty} $ (in equation  \ref{Paper_II_Majorization_Definition} {take} $k =1 $). \\
		
		\noindent Therefore, for {the} simple kriging from {Equation~}\eqref{Paper_II_SMSPE_SK} we can say $\displaystyle{SMSPE_{sk}(\norm{\VectD}_{\infty} ) \leq SMSPE_{sk}(\norm{\VectD^{\prime}}_{\infty} ) }$ as $\Big( 1-  s_{\alpha}(d) \; s_{\beta}(\delta) \Big) $ is increasing in $d$ and $\delta$.\\
		
		Similarly, for ordinary kriging model, in the expression of $SMSPE_{ok}$ (from equation \ref{Paper_II_SMSPE_OK}), we have $\displaystyle{\Big( 1-  s_{\alpha}(\norm{\VectD}_{\infty} 
			) \; s_{\beta}(\norm{\VectDelta}_{\infty}) \Big)  \leq \Big( 1-  s_{\alpha}(\norm{\VectD^{\prime}}_{\infty} ) \; s_{\beta}(\norm{\VectDelta^{\prime}}_{\infty}) \Big)  }$ and  \\ $\displaystyle{\Big( 1- u_{\alpha}(\norm{\VectD}_{\infty} ) u_{\beta}(\norm{\VectDelta}_{\infty})\Big)^2 \leq  \Big( 1- u_{\alpha}(\norm{\VectD^\prime}_{\infty} ) u_{\beta}(\norm{\VectDelta^\prime}_{\infty})\Big)^2 } $ as $\displaystyle{\Big( 1-  s_{\alpha}(d) \; s_{\beta}(\delta) \Big) }$ and $\displaystyle{\Big( 1- u_{\alpha}(d) u_{\beta}(\delta)\Big)^2 }$ are increasing in $d$ and $\delta$. Also, $\displaystyle{\dfrac{1}{\Omega_{x}(\xibold )}}$ and $\displaystyle{\dfrac{1}{\Omega_{y}(\xibold )}}$ {are} Schur-convex in $d_{i}$ and $\delta_{j}$, respectively. Therefore, $\displaystyle{\dfrac{1}{\Omega_{x}(\xibold)} \leq \dfrac{1}{\Omega_{x}(\xibold^\prime)}}$ and $\displaystyle{\dfrac{1}{\Omega_{y}(\xibold)} \leq \dfrac{1}{\Omega_{y}(\xibold^\prime)}}$. So, we can say $\displaystyle{SMSPE_{ok}(\xibold) \leq SMSPE_{ok}(\xibold^\prime)}$. 
	\end{proof}
	
	\section{Appendix} 	 \label{AppendixG1}
	The {expressions} for $SMSPE_{sk}$ and $SMSPE_{ok}$ {in equations} \eqref{Paper_II_SMSPE_SK} and  \eqref{Paper_II_SMSPE_OK} are symmetric and contain separable terms in the $x$-covariates and $y$-covariates (equivalently $d_{i}$ and $\delta_{j}$). Hence, if we work with only the $x$-covariates to determine the relationship between the design and $SMSPE$, the same argument follows for the $y$-covariates. 
	
	Let us see some notations that will be useful for proving Results~\ref{AppendixG1_Result1} and \ref{AppendixG1_Result2}. For the new design $\xiboldPlus \equiv (\VectDPlus, \VectDeltaPlus)$, without loss of generality assume that design points $\displaystyle{\{x_{1}^{\prime}, \ldots, x_{n_1}^{\prime} \}}$ and $\displaystyle{\{y_{1}^{\prime}, \ldots, y_{m_1}^{\prime} \}}$ are added between $(x_{i}, x_{i+1})$ and $(y_{j}, y_{j+1})$ for some $i=1,\ldots,n$ and $j=1,\ldots,m$, respectively.  Relabel $x_{i} \equiv x_{i_0}^{\prime} $, $x_{l}^{\prime} \equiv x_{i_l}^{\prime}$ for $l=1,\ldots,n_{1}$,  and $x_{i+1} \equiv x_{i_{ n_{1}+1}}^{\prime}$. Let $d_{l}^{\prime} = x_{i_{l+1}}^{\prime} - x_{i_l}^{\prime}$ for $l=0,\ldots,n_{1}$. So we have $\displaystyle{ \VectDPlus = (  d_{1}, \ldots,d_{i-1}, d_{0}^{\prime}, \ldots, d_{n_{1}}^{\prime}, d_{ i+1} , \ldots, d_{ n-1}) }$. Similarly, relabel $y_{j} \equiv y_{j_0}^{\prime} ,   y_{k}^{\prime} \equiv y_{j_k}^{\prime}$ for $k=1,\ldots,m_{1}$, and $y_{j+1} \equiv y_{j_{ m_{1}+1}}^{\prime}$. Let $\delta_{k}^{\prime} = y_{j_{k+1}}^{\prime} - y_{j_k}^{\prime}$ for $k=0,\ldots,m_{1}$. Then, $ \VectDelta= ( \delta_{1}, \ldots,\delta_{j-1},  \delta_{0}^{\prime}, \ldots, \delta_{m_{1}}^{\prime}, \delta_{ j+1}, \ldots, \delta_{ m-1})$.

	\begin{result} \label{AppendixG1_Result1}
	For the simple kriging model, in the expression for $SMSPE_{sk}$ given by equation \eqref{Paper_II_SMSPE_SK}, we have seen $\Big( 1-  s_{\alpha}(d) \; s_{\beta}(\delta) \Big) $ is increasing in $d$ and $\delta$. As, $\norm{\VectDPlus}_{\infty} \leq \norm{\VectD}_{\infty}$ and $ \norm{\VectDeltaPlus}_{\infty}  \leq \norm{\VectDelta}_{\infty}$, therefore $SMSPE_{sk}(\xiboldPlus) \leq SMSPE_{sk}(\xibold)$. Note, that for a design having $n$ points for {the $x$-covariate,} if we take the case where $d_{i} = \dfrac{1}{n-1}$ for all $i$ and the number of points added is less than $n-1$, then $SMSPE_{sk}$ remains unchanged. In this case we need to add {at least} $n-1$ additional points to reduce the $SMSPE_{sk}$. 		
	\end{result}
	
	\begin{result} \label{AppendixG1_Result2}
	For {the} ordinary kriging model, in the expression {for} $SMSPE_{ok}$ in equation \eqref{Paper_II_SMSPE_OK}, we know that the two terms $\Big( 1-  s_{\alpha}(d) \; s_{\beta}(\delta) \Big) $ and $\Big( 1- u_{\alpha}(d) u_{\beta}(\delta)\Big)^2 $ are increasing in $d$ and $\delta$. Also, $\norm{\VectDPlus}_{\infty} \leq \norm{\VectD}_{\infty}$ and $ \norm{\VectDeltaPlus}_{\infty}  \leq \norm{\VectDelta}_{\infty}$. So, $\displaystyle{\Big( 1-  s_{\alpha}(\norm{\VectDPlus}_{\infty} 
	) \; s_{\beta}(\norm{\VectDeltaPlus}_{\infty}) \Big)  \leq \Big( 1-  s_{\alpha}(\norm{\VectD}_{\infty} ) \; s_{\beta}(\norm{\VectDelta}_{\infty}) \Big)  }$ and  $\displaystyle{\Big( 1- u_{\alpha}(\norm{\VectDPlus}_{\infty} ) u_{\beta}(\norm{\VectDeltaPlus}_{\infty})\Big)^2 \leq  \Big( 1- u_{\alpha}(\norm{\VectD}_{\infty} ) u_{\beta}(\norm{\VectDelta}_{\infty})\Big)^2 } $. In this case it is necessary to check only if $\displaystyle{ \dfrac{1}{\Omega_{x}(\xiboldPlus) } < \dfrac{1}{\Omega_{x}(\xibold) } }$ and $\displaystyle{\dfrac{1}{\Omega_{y}(\xiboldPlus)}  < \dfrac{1}{\Omega_{y}(\xibold)}}$. Here, let $\MatrixPHat$ be the the exponential correlation matrix with parameter $\alpha$ corresponding to the set $\displaystyle{\{x_1, \ldots, x_i, x_{1}^{\prime}, \ldots, x_{n_1}^{\prime}, x_{i+1}, \ldots, x_n \}}$. So we have 
	\begin{align}
	\Omega_{x}(\xiboldPlus) - \Omega_{x}(\xibold) &=   \VectOneNnOne^{T} \MatrixPHat^{-1}  \VectOneNnOne - \VectOneN^{T} \MatrixP^{-1}  \VectOneN \nonumber \\
	&=  \sum_{l=0}^{n_{1}}\dfrac{e^{\alpha d_{l}^{\prime}} - 1}{e^{\alpha d_{l}^{\prime}} + 1} - \dfrac{e^{\alpha d_{i}} - 1}{e^{\alpha d_{i}} + 1} \label{Paper_II_LocEq1} \;\;\;\;\;\;\;\;\;\;\;\;\;\;\; \text{  where, } \sum_{l=0}^{n_{1}} d_{l}^{\prime} = d_{i}
	\end{align} 
	Define $v_{\theta}(.)$ over $\mathbb{R}^{+}$ such that $v_{\theta}(\zeta) = \dfrac{e^{\alpha \zeta} - 1}{e^{\alpha \zeta} + 1}$. We want to see if $\sum_{l=0}^{n_{1}} v_{\alpha}(d_{l}^{\prime}) - v_{\alpha}(d_{i}) \geq 0$. First consider $n_{1} = 1$, then 
	\begin{align}
	\sum_{l=0}^{1} v_{\alpha}(d_{l}^{\prime}) - v_{\alpha}(d_{i})
	&=  \dfrac{e^{\alpha d_{0}^{\prime}} - 1}{e^{\alpha d_{0}^{\prime}} + 1} + \dfrac{e^{\alpha d_{1}^{\prime}} - 1}{e^{\alpha d_{1}^{\prime}} + 1} - \dfrac{e^{\alpha d_{i}} - 1}{e^{\alpha d_{i}} + 1} 
	\;\;\;\;\;\;\;\;\;\;\;\;\;\;\; \text{  where, }  d_{0}^{\prime} +  d_{1}^{\prime}= d_{i} \nonumber \\
	&= \dfrac{(e^{\alpha d_{i}} - 1) (1 + e^{\alpha d_{i}} - e^{\alpha d_{0}^{\prime}} - e^{\alpha d_{1}^{\prime}})}{(e^{\alpha d_{0}^{\prime}} + 1) (e^{\alpha d_{1}^{\prime}} + 1)  (e^{\alpha d_{i}} + 1)} \nonumber \\
	&= \dfrac{(e^{\alpha d_{i}} - 1) (e^{\alpha d_{1}^{\prime}} - 1 ) (e^{\alpha d_{0}^{\prime}} - 1 )}{(e^{\alpha d_{0}^{\prime}} + 1) (e^{\alpha d_{1}^{\prime}} + 1)  (e^{\alpha d_{i}} + 1)}  > 0 \label{Paper_II_LocEq2}.
	\end{align} 
	After this step it is very easy to check using induction, that $\Omega_{x}(\xiboldPlus) - \Omega_{x}(\xibold) > 0$ for any integer $n_{1} \geq 1$. Therefore, $ \dfrac{1}{\Omega_{x}(\xiboldPlus) } < \dfrac{1}{\Omega_{x}(\xibold) } $ and similarly $\dfrac{1}{\Omega_{y}(\xiboldPlus)}  < \dfrac{1}{\Omega_{y}(\xibold)}$. So, it is clear that with subsequent addition of points over any covariate axis of the grid, the $SMSPE_{ok}$ reduces. Therefore, depending on the {budget,} an experimenter can add as many sample points as possible to an existing design for an ordinary kriging situation.		
	\end{result}

	\section{Appendix} 		\label{AppendixH}
	Proof of Lemma~\ref{Paper_II_Retrospective_Lemma_I}. 
	
	We show that for simple and ordinary kriging {models}, if $n_{1}$ and $m_{1}$ new points are added between $(x_{i}, x_{i+1})$ and $(y_{j}, y_{j+1})$ for some $i,j${, then equally spacing these new points within the intervals $(x_{i}, x_{i+1})$ and $(y_{j}, y_{j+1})$} will lead to {the} minimum possible value {of} $SMSPE$. 
	\begin{proof}
				As before we denote the initial design by $\xibold$ and the new design after adding new points by $\xiboldPlus$. It is very easy to see, if the new points are equally spaced (evenly spaced) between $(x_{i}, x_{i+1})$ and $(y_{j}, y_{j+1})$, then this will lead to {the} minimum possible {values} of $\norm{\VectDPlus}_\infty$ and $\norm{\VectDeltaPlus}_\infty$. 
			
			In {the} case of {the} simple kriging model, see equation~\eqref{Paper_II_SMSPE_SK} for $SMSPE_{sk}$. The expression $\Big( 1-  s_{\alpha}(\norm{\VectDPlus}_\infty) \; s_{\beta}(\norm{\VectDeltaPlus}_\infty) \Big) $ is increasing in its variables{;} it attains minimum possible value when $\norm{\VectDPlus}_\infty$ and $\norm{\VectDeltaPlus}_\infty$ {attain} their minimum value. So, $SMSPE_{sk}(\xiboldPlus)$ is minimized if the new points are equally spaced (or evenly spaced) between $(x_{i}, x_{i+1})$ and $(y_{j}, y_{j+1})$. 
			
			For ordinary kriging{,} see equation~\eqref{Paper_II_SMSPE_OK} for $SMSPE_{ok}$. The expressions, $\Big( 1- u_{\alpha}(\norm{\VectDPlus}_\infty) u_{\beta}(\norm{\VectDeltaPlus}_\infty)^2 $ and $\Big( 1-  s_{\alpha}(\norm{\VectDPlus}_\infty) \; s_{\beta}(\norm{\VectDeltaPlus}_\infty) \Big) $ {are} increasing, {hence the minimum possible values are attained} when $\norm{\VectDPlus}_\infty$ and $\norm{\VectDeltaPlus}_\infty$ {attain their minimum values}. Also, for new design $\xiboldPlus$, we have $\displaystyle{\Omega_{x}(\xiboldPlus) = \sum_{    \substack{j=1,\ldots,i-1\\\;\;\;\;\; i+1,\ldots,n-1} } \Big( d_{j} +   \dfrac{e^{\alpha d_{j}} - 1}{e^{\alpha d_{j}} + 1}  \Big)  + \sum_{l=0}^{n_{1}} \Big( d_{l}^{\prime} + \dfrac{e^{\alpha d_{l}^{\prime}} - 1}{e^{\alpha d_{l}^{\prime}} + 1}  \Big)}$; $\displaystyle{\Bigg(  \sum_{l=0}^{n_{1}} \Big( d_{l}^{\prime} + \dfrac{e^{\alpha d_{l}^{\prime}} - 1}{e^{\alpha d_{l}^{\prime}} + 1}\Big) \Bigg) ^{-1}}$ is {a} Schur-convex function in variable $d_{l}^{\prime}$ (proved in \citealp{Paper_One_2021_arxiv}) and {is} minimized when {the} $d_{l}^{\prime}$'s are equal. Therefore, $\dfrac{1}{\Omega_{x}(\xiboldPlus)}$ attains minimum value if the $n_{1}$ new points are spaced equally between $(x_{i}, x_{i+ 1} )$ and similarly $\dfrac{1}{\Omega_{y}(\xiboldPlus) }$ attains minimum value if the $m_{1}$ new points are spaced equally between $(y_{j}, y_{j+ 1} )$. So, $SMSPE_{ok}(\xiboldPlus)$ is minimized if the new points are equally spaced (or evenly spaced) between $(x_{i}, x_{i+1})$ and $(y_{j}, y_{j+1})$. 		
	\end{proof}

			\section{Appendix} 		\label{AppendixG2}
			Proof of correctness of Algorithm~\ref{PaperII_Algorithm1}.
					\begin{proof}
				We add one point at a time, {say} $x_{1}^{\prime}$ or (and) $y_{1}^{\prime}$ (to $x$ or (and) $y$-covariate){. We need to ensure at the} end of Step 2 the design obtained is the best possible design. 
				
				If we add $x_{1}^{\prime}$ or (and) $y_{1}^{\prime}$ between $[x_{i_{0}}, x_{i_{0} + 1} ]$ or (and) $[y_{j_{0}}, y_{j_{0} + 1} ]$ for some {$i_{0}= 1,\ldots, length(\VectD)$ and $j_{0}= 1,\ldots,length(\VectDelta) $} then, {by Lemma~\ref{Paper_II_Retrospective_Lemma_I},} $x_{1}^{\prime}$ and $y_{1}^{\prime}$ should be added at the {mid-point} of the respective intervals. 
				
				Now, we are left with identifying the interval at which the points should be inserted (that is{, the choice of} $i_{0}$ and $j_{0}$). 
				
				For simple kriging, from equation~\eqref{Paper_II_SMSPE_SK}, we can see that $\Big( 1-  s_{\alpha}(\norm{\VectDPlus}_\infty) \; s_{\beta}(\norm{\VectDeltaPlus}_\infty) \Big)  $ attains minimum possible value when $i_{0}$ and $j_{0}$ are chosen such that $d_{i_{0}} = \norm{\VectD}_\infty $ and $ \delta_{j_{0}} = \norm{\VectDelta}_\infty$. Hence after Step 2 in Algorithm~\ref{PaperII_Algorithm1} we get {the} best possible design for {the} simple kriging model. 
				
				For ordinary kriging, in equation \eqref{Paper_II_SMSPE_OK}, $\Big( 1-  s_{\alpha}(\norm{\VectDPlus}_\infty) \; s_{\beta}(\norm{\VectDeltaPlus}_\infty) \Big)  $ and $\Big( 1- u_{\alpha}(\norm{\VectDPlus}_\infty) u_{\beta}(\norm{\VectDeltaPlus}_\infty)\Big)^2 $ {attain} minimum possible value when $i_{0}$ and $j_{0}$ are chosen such that $d_{i_{0}} = \norm{\VectD}_\infty $ and $ \delta_{j_{0}} = \norm{\VectDelta}_\infty$. Next, let us see how {the} minimum possible value for $\dfrac{1}{\Omega_{x}(\xiboldPlus_{a1})} $ can be ensured. Here, let $\MatrixPHat$ be the the exponential correlation matrix with parameter $\alpha$ corresponding to the set $\displaystyle{\{x_1, \ldots, x_{i_0}, x_{1}^{\prime}, x_{i_0+1}, \ldots, x_n \}}$. If  $x_{1}^{\prime}$ is inserted at the {mid-point} of $[x_{i_{0}}, x_{i_{0} + 1} ]$ then, relabeling the new design to be $\xiboldPlus_{a1}$ we get, 
				\begin{align}
				\mathcal{I}(d_{i_{0}}) &=  \Omega_{x}(\xiboldPlus_{a1}) - \Omega_{x}(\xibold) \nonumber \\
				\mathcal{I}(d_{i_{0}}) &=  \VectOneNnOne^{T} \MatrixPHat^{-1}  \VectOneNnOne - \VectOneN^{T} \MatrixP^{-1}  \VectOneN \nonumber \\
				&=  2 \;\; \dfrac{e^{\alpha d_{i_{0}}/2} - 1}{e^{\alpha d_{i_{0}}/2} + 1} - \dfrac{e^{{\alpha d_{i_{0}}}} - 1}{e^{\alpha d_{i_{0}}} + 1} \label{PaperII_Loc1}.
				\end{align}
				Also,
				\begin{align*}
				\dfrac{\partial \;\; \mathcal{I}(d)}{\partial d} 
				&= \dfrac{2 \alpha e^{\alpha d/2}}{(e^{\alpha d/2}+1)^2} - \dfrac{2 \alpha e^{\alpha d}}{(e^{\alpha d}+1)^2} \\
				&= \dfrac{2 \alpha e^{\alpha d/2} (e^{\alpha d/2}-1) (e^{3 \alpha  d/2}-1) }{(e^{\alpha d/2}+1)^2\; (e^{\alpha d}+1)^2} > 0
				\end{align*}
				Hence, $\mathcal{I}(d_{i_{0}})$ is an increasing function of $d_{i_{0}}$, which means $\Omega_{x}(\xiboldPlus_{a1}) $ attains maximum possible value when $d_{i_{0}} = \norm{\VectD}_\infty$, that is, $i_{0}$ is chosen in such a way $[x_{i_{0}}, x_{i_{0} + 1} ]$ is the largest interval. Hence, we can say Step 2 of Algorithm~\ref{PaperII_Algorithm1} {ensure} best possible designs at each step for ordinary kriging. 
				
			\end{proof}

	\section{Appendix} \label{Paper_II_AppendixI}
	Proof of Theorem~\ref{Paper_II_Theorem_Loc1}.
	\begin{proof}
			Take the initial design to be $\xibold$ as before. Suppose grid design $\xiboldPlus_{a1} \equiv (\VectDPlus_{a1}, \VectDeltaPlus_{a1})$ is obtained by adding $n_{1_0}^{(i)}$ and $m_{1_0}^{(j)}$ number of points between $(x_{i}, x_{i+1})$ and $(y_{j}, y_{j+1})$, respectively, using Algorithm~\ref{PaperII_Algorithm1} such that $\displaystyle{\sum_{i=1}^{n-1} n_{1_0}^{(i)} = n_{1}}$ and $\displaystyle{\sum_{j=1}^{m-1} m_{1_0}^{(j)} = m_{1}}$. In this case, the $n_{1_0}^{(i)}$ and $m_{1_0}^{(j)}$ new points are not necessarily {equally spaced within} $(x_{i}, x_{i+1})$ and $(y_{j}, y_{j+1})$, respectively. Note that in Algorithm~\ref{PaperII_Algorithm1}, at each stage a point is placed at the {mid-point} of the largest gap, so {the final design will not necessarily have all points equally spaced} between existing intervals $(x_{i}, x_{i+1})$ or $(y_{j}, y_{j+1})$ for some $i,j$. For example, in Figure \ref{Fig_Ret_Illus1_1}, the two points placed on {the} $y$-axis {in the} interval $(0.4, 1)$ are not {equispaced}. 
	
	Now, consider the design $\xibold^{+}_{k_{o}l_{o}} \equiv ({\pmb d^{+}_{k_{o}}},{\pmb \delta^{+}_{l_{o}}}) \in \SetUTwoPrime${,} which is obtained by putting $n_{1_0}^{(i)}$ and $m_{1_0}^{(j)}$ points between $(x_{i}, x_{i+1})$ and $(y_{j}, y_{j+1})$, respectively, in an equally spaced manner {using} Algorithm~\ref{PaperIIAlgorithm_AdditionChoiceSet} (as Algorithm~\ref{PaperII_Algorithm2} utilizes Algorithm~\ref{PaperIIAlgorithm_AdditionChoiceSet} to construct the initial choice set). Denote the partition vector of $(x_{i}, x_{i+1})$ obtained by running  Algorithm~\ref{PaperII_Algorithm1} by $\pmb d_{i_{a1}}$ and the partition vector obtained by running Algorithm~\ref{PaperIIAlgorithm_AdditionChoiceSet} by $\pmb d_{i_{a2}}$ where $n_{1_0}^{(i)}$ equispaced points are added between $(x_{i}, x_{i+1})$, then $\pmb d_{i_{a2}} \prec \pmb d_{i_{a1}}$ for any $i = 1, \ldots, n-1	$, as components of $\pmb d_{i_{a2}}$ are equal. Clearly, the components of $\VectDPlus_{a1}$ majorize the corresponding components of $\pmb d^{+}_{k_{o}} $. Therefore, Proposition A.7 from \cite{MarshalOlkinBook} is used to conclude $     \pmb d^{+}_{k_{o}} \prec  \VectDPlus_{a1}$. Similarly, it can be argued that $ \pmb \delta^{+}_{l_{o}} \prec   \VectDeltaPlus_{a1}  $. 
	
	Hence, using Theorem~\ref{Paper_II_Theorem_Any_Design}, $SMSPE(\xibold^{+}_{k_{o}l_{o}}) \leq SMSPE(\xiboldPlus_{a1})$. That means, there is at least one design in $\SetUTwoPrime$ which is no worse than the partition obtained by Algorithm~\ref{PaperII_Algorithm1}.
	
	As, $ \xiboldPlus_{a2} \equiv ({\VectDPlus_{a2}};{\VectDeltaPlus_{a2}}) $ is the best possible design obtained using Algorithm~\ref{PaperII_Algorithm2}, therefore $ SMSPE(\xiboldPlus_{a2}) \leq SMSPE(\xibold^{+}_{k_{o}l_{o}}) $. So, $SMSPE(\xiboldPlus_{a2}) \leq SMSPE(\xiboldPlus_{a1})$
	and hence $\xiboldPlus_{a2}$ is at least as good as $\xiboldPlus_{a1}$. 
	
	\end{proof}

	\end{appendix}
		

		\bibliography{PaperII}       	
		
\end{document}